\documentclass[acmsmall]{acmart}

\AtBeginDocument{%
  \providecommand\BibTeX{{%
    \normalfont B\kern-0.5em{\scshape i\kern-0.25em b}\kern-0.8em\TeX}}}

\usepackage{amsmath}
\usepackage{multirow}
\usepackage{graphicx}
\usepackage{textcomp}
\renewcommand{\vec}[1]{\boldsymbol{#1}}
\usepackage{setspace}
\usepackage{subfigure}

\usepackage{algorithm}
\usepackage{algorithmic}
\usepackage{color}

\newtheorem{thm}{Theorem}
\newtheorem{lem}{Lemma}

\newtheorem{defn}{Definition}

\newtheorem{rem}{Remark}
\newtheorem{pro}{Problem}

\begin{document}

\title{Adaptive Influence Maximization: If Influential Node Unwilling to Be the Seed}

\author{Jianxiong Guo}
\email{jianxiong.guo@utdallas.edu}
\authornote{Corresponding author}
\affiliation{%
  \institution{Department of Computer Science, The University of Texas at Dallas}
  \streetaddress{800 W Campbell Rd}
  \city{Richardson}
  \state{Texas}
  \country{USA}
  \postcode{75080}
}

\author{Weili Wu}
\email{weiliwu@utdallas.edu}
\affiliation{%
	\institution{Department of Computer Science, The University of Texas at Dallas}
	\streetaddress{800 W Campbell Rd}
	\city{Richardson}
	\state{Texas}
	\country{USA}
	\postcode{75080}
}

\renewcommand{\shortauthors}{J. Guo et al.}

\begin{abstract}
Influence maximization problem attempts to find a small subset of nodes that makes the expected influence spread maximized, which has been researched intensively before. They all assumed that each user in the seed set we select is activated successfully and then spread the influence. However, in the real scenario, not all users in the seed set are willing to be an influencer. Based on that, we consider each user associated with a probability with which we can activate her as a seed, and we can attempt to activate her many times. In this paper, we study the adaptive influence maximization with multiple activations (Adaptive-IMMA) problem, where we select a node in each iteration, observe whether she accepts to be a seed, if yes, wait to observe the influence diffusion process; If no, we can attempt to activate her again with a higher cost or select another node as a seed. We model the multiple activations mathematically and define it on the domain of integer lattice. We propose a new concept, adaptive dr-submodularity, and show our Adaptive-IMMA is the problem that maximizing an adaptive monotone and dr-submodular function under the expected knapsack constraint. Adaptive dr-submodular maximization problem is never covered by any existing studies. Thus, we summarize its properties and study its approximability comprehensively, which is a non-trivial generalization of existing analysis about adaptive submodularity. Besides, to overcome the difficulty to estimate the expected influence spread, we combine our adaptive greedy policy with sampling techniques without losing the approximation ratio but reducing the time complexity. Finally, we conduct experiments on several real datasets to evaluate the effectiveness and efficiency of our proposed policies.
\end{abstract}




\maketitle

\section{Introduction}
Online social networks (OSNs) were blossoming prosperously in recent decades and have become the main means of communication between people such as Wechat, Facebook, Twitter, and LinkedIn. More and more people participate to discuss the topics that they are interested in on these social platforms. Many companies or advertisers exploit the relations established in OSNs to spread their products, opinions, or innovations. They provide those influential individuals (called ``seed nodes'') with free or discounted samples, in order to create widespread influence across the whole network via word-of-mouth effect \cite{domingos2001mining} \cite{richardson2002mining}. Based on that, influence maximization (IM) problem \cite{kempe2003maximizing} was formulated, which selects a subset of users (called ``seed set'') for an information cascade to maximize the expected follow-up adoptions (influence spread). It is a general model for a number of realistic scenarios, such as viral marketing. In \cite{kempe2003maximizing}, they created two discrete influence diffusion models, independent cascade model (IC-model) and linear threshold model (LT-model), where IC-model relies on peer-to-peer communication but LT-model computes the total influence from user’s neighbors. Then, they proved that the IM problem is NP-hard and provided a $(1-1/e)$-approximate algorithm by simple the greedy strategy in the framework of submodularity. After this groundbreaking work, a sequence of derivative problems appeared and were solved under the different constraints and scenarios \cite{chen2011influence} \cite{chen2016robust} \cite{guo2019novel} \cite{guo2020influence}, such as profit maximization \cite{guo2020k}, community partition, and rumor blocking (detection).

Despite these developments, the existing researches on the IM problem have a crucial drawback. When selecting a seed set at the beginning, they all seem to take it for granted that every user in their selected seed set can be activated successfully to become an active seed and then spread the influence as they wish. However, there are some users unwilling, even impossible, to be the influencers. For example, the hottest topic at the moment, Coronavirus in Wuhan, China. Some non-profit organizations or official media are trying to make celebrities speak out to ease the panic among the people. However, due to self-interest or other factors, some celebrities are not willing to be their ``seed nodes''. Then, we have two options, one is trying to persuade those who stand on the opposite side sentimentally and rationally, the other is giving up and look for other potential ``seed nodes''. Based on that, we design a new IM with multiple activations (IMMA) problem, where a node can be activated to be an influencer with a probability when we select it as a seed and we can attempt to activate it many times. For the same node, each attempt is referred to as a ``trial'' and each trial has a cost. If the first trial fails, we can conduct the second trial, the third trial, etc., but their cost is higher than the first.

Most existing techniques on the IM problem concentrate on non-adaptive strategies, which are required to select all seed nodes at once without observing actual node status and diffusion process. In other words, we need to point out the seed set and the number of trials for each node in this seed set in one batch. As a result, it may return a seed that cannot be activated actually or assign too many trials to this node. For example, we give a seed three trials and it is activated in the first trial, so the remaining two waste our budget. Thus, the non-adaptive seeding strategy is not the best choice to solve our IMMA problem. Golovin et al. \cite{golovin2011adaptive} studies the IM problem under the adaptive strategy, they select the $(i+1)$-th node after observing the influence diffusion of the first $i$ nodes until all seed nodes are chosen. Based on that, the Adaptive-IMMA problem is proposed in this paper, which selects a seed and attempts to activate it in each iteration. If successful, wait to observe its influence process; If failed, record this failed trial. Those nodes on which all the trials are unsuccessful can be considered as seed again in a later step. 

Golovin \textit{et al.} \cite{golovin2011adaptive} provided a $(1-1/e)$-approximate algorithm by an adaptive greedy policy for the adaptive IM problem in the framework of adaptive monotonicity and adaptive submodularity. However, for our Adaptive-IMMA problem, its solution is a seeding vector $\vec{x}\in\mathbb{Z}^V_+$, not a seed set $S\subseteq V$, where each component $\vec{x}(u)$ means how many activation attempts we give to user $u$. It will be executed sequentially. For example, $\vec{x}(u)=i$, we will do trial $\langle u,1\rangle$, trial $\langle u,2\rangle$ until trial $\langle u,i\rangle$ on user $u$. The domain of the objective function is defined on integer lattice, not generally on set. Thus, traditional analytical methods based on adaptive submodularity cannot be applied to analyze our problem. The submodularity shows us with diminishing marginal gain property in set function. For functions defined on integer lattice, such a property exists as well, called dr-submodularity. Based on that, we define the concepts of adaptive monotonicity on integer lattice and adaptive dr-submodularity formally, which are extended from adaptive submodularity on set function \cite{golovin2011adaptive} and dr-submodularity on integer lattice \cite{soma2015generalization}. Then, we formulate the objective function of our Adaptive-IMMA problem and prove it is adaptive monotone on integer lattice as well as adaptive dr-submodular. Each trial $\langle u,i\rangle$ is associated with a cost $c(\langle u,i\rangle)$ and the costs of different trials are different. Given a randomized policy $\pi(\kappa)$, the total budget $k$ is an expected knapsack constraint such that the total cost of the seeding vector returned by $\pi(\kappa)$ is less than $k$ expectedly. Then, we study the approximate performance of adaptive greedy policy for maximizing adaptive monotone and adaptive dr-submodular functions under the expected knapsack constraint, which is a non-trivial generalization of existing analysis about the approximate performance of adaptive submodular functions. Assume $c(\langle u,i\rangle)\leq c(\langle u,i+1\rangle)$, it returns an acceptable solution with $(1-1/e)$-approximate ratio.

Besides, it is \#P-hard to compute the expected influence spread given a seed set under the IC-model \cite{chen2010scalable} and LT-model \cite{chen2010scal}. The complexity of our objective function is higher. In order to overcome this shortcoming, we design an unbiased estimator of the conditional expected marginal gain for our problem based on the reverse influence sampling (RIS) \cite{borgs2014maximizing}. Adapted from state-of-the-art EPIC algorithm \cite{huang2020efficient} for the IM problem, combined it with our adaptive greedy policy, we formulate a sampled adaptive greedy policy and achieve a $(1-\exp(-1+\varepsilon))$ expected approximation guarantee. Its time complexity is reduced significantly. Finally, we conduct several experiments to evaluate the superiority of our adaptive policies over their corresponding non-adaptive algorithms and the sampled adaptive greedy policy over other heuristic adaptive policies, which support the effectiveness and efficiency of our approaches strongly.

\section{Related Work}
\textbf{Influence maximization (IM):} The IM problem has been studied extensively. Kempe \textit{et al.} \cite{kempe2003maximizing} formulated IM as a combinatorial optimization problem, proposed the triggering model, including IC-model and LT-model, and gave us a greedy algorithm with $(1-1/e-\varepsilon)$-approximation. It was implemented by Monte-Carlo (MC) simulations with high time complexity. Chen \textit{et al.} \cite{chen2010scalable} \cite{chen2010scal} followed Kempe's work and proved its \#P-hardness to compute the expected influence spread. Thus, the running time was too slow to apply to larger real networks. After those seminal works, a lot of researchers made an effort to improve its time complexity. Brogs \textit{et al.} \cite{borgs2014maximizing} proposed the concept of reverse influence sampling (RIS) to estimate the expected influence spread, which is scalable in practice and guaranteed theoretically at the same time. Then, a series of more efficient randomized algorithms were arisen, such as TIM/TIM+ \cite{tang2014influence}, IMM \cite{tang2015influence}, SSA/D-SSA \cite{nguyen2016stop}, and OPIM-C \cite{tang2018online}. Recently, Han \textit{et al.} \cite{han2018efficient} provided us with an EPIC algorithm with an expected approximation guarantee. Then, the EPIC was improved further based on the OPIM-C \cite{tang2018online}, which is the most efficient algorithm to solve the IM problem until now \cite{huang2020efficient}. They were scalable algorithms for the IM problem and can be adapted to other relative problems. However, all of these are used to solve the IM problem under the non-adaptive setting.

\textbf{Dr-submodular maximization and its applications in social networks:} Defined on integer lattice, dr-submodular maximization problem is a hot topic that attracts a lot of researchers recently. Soma \textit{et al.} \cite{soma2015generalization} formalized dr-submodularity on integer lattice inspired by the diminishing return property on set and addressed a submodular cover problem. Soma \textit{et al.} \cite{soma2018maximizing} studied the monotone dr-submodular maximization problem on integer lattice systematically, where they proposed a series of $(1-1/e-\varepsilon)$-approximate algorithms for the cardinality constraint, the knapsack constraint, and the polymatroid constraint. Applied to solve problems social networks, Chen \textit{et al.} \cite{chen2020scalable} gave a lattice IM problem defined on integer lattice, whose objective function is monotone and dr-submodular. Then, they proposed azi scalable algorithm with $(1-1/e-\varepsilon)$ approximation ratio that is adapted from the IMM \cite{tang2015influence}. Guo \textit{et al.} \cite{guo2020continuous} proposed a continuous activity maximization problem and provided a solution framework for maximizing a monotone but not dr-submodular function by the sandwich approximation approach. Other literature and results about dr-submodular maximization and its applications are shown in \cite{gottschalk2015submodular} \cite{soma2017non} \cite{hatano2016adaptive}.

\textbf{Adaptive influence maximization:} Golovin \textit{et al.} \cite{golovin2011adaptive} extended the submodularity to adaptive settings and obtained the same approximation ratio for the adaptive IM problem because its objective function is adaptive monotone and adaptive submodular under the full-adoption feedback model where only one node can be selected in each iteration. However, the objective function of the adaptive IM problem is not adaptive submodular under the myopic feedback model \cite{golovin2011adaptive} or the partial feedback model \cite{yuan2017no} \cite{tang2020influence}. Tong \textit{et al.} \cite{tong2020adaptive} gave us a systematic framework about the adaptive IM problem when it is not adaptive submodular, where they designed a seeding strategy and showed the approximation analysis by introducing the concept of regret ratio. Unfortunately, all these works were based on a fact that the expected influence spread can be computed accurately in polynomial time, which is an unrealistic assumption. To improve its scalability, Han \textit{et al.} \cite{han2018efficient} proposed an AdaptGreedy framework instantiated by scalable IM algorithms to solve the adaptive IM problem where it can select a batch of seed nodes in each iteration, which returns a worst-case approximation guarantee with high probability. Sun \textit{et al.} \cite{sun2018multi} studied a multi-round influence maximization (MRIM) problem where information diffuses in multiple rounds independently from different seed sets. They considered MRIM problem under the adaptive setting and designed an adaptive algorithm with $(1-\exp(1/e-1)-\varepsilon)$ approximation instantiated by the IMM \cite{tang2015influence}. Tang \textit{et al.} \cite{tang2019efficient} considered the seed minimization problem under the adaptive setting and proposed an ASTI algorithm that offers an expected approximation ratio. Recently, Huang \textit{et al.} \cite{huang2020efficient} pointed out there are some mistakes in the approximation analysis of adaptive policies in \cite{han2018efficient} \cite{sun2018multi}. They fixed the previous AdaptGreedy framework in \cite{han2018efficient} and proved it has a $(1-\exp(\rho_b(\varepsilon-1)))$ expected approximation guarantee instantiated by their improved EPIC in \cite{huang2020efficient}.

Neverthelessbut, none of them considered a problem that is adaptive dr-submodular, especially under the expected knapsack constraint. This is the main contribution in this paper.

\section{Problem Formulation}
In this section, we define the problem of our adaptive influence maximization on multiple activations formally and introduce some preliminary knowledges.

\subsection{Influence Model and Graph Realization}
A social network can be denoted by a directed graph $G=(V,E)$ where $V=\{v_1,v_2,\cdots,v_n\}$ is the node (user) set with $|V|=n$, $E=\{e_1,e_2,\cdots,e_m\}$ is the edge set with $|E|=m$, which describes the relationship between users. For any edge $e=(u,v)\in E$, $v$ is an outgoing neighbor of $u$ and $u$ is an incoming neighbor of $v$. For any node $v\in N$, we denote by $N^-(v)$ its set of incoming neighbors and $N^+(v)$ its set of outgoing neighbors. Each edge $(u,v)$ is associated with a diffusion probability $p_{uv}\in (0,1]$. Thus, the influence diffusion on this network is stochastic.

Let $S\subseteq V$ be a given node set, the influence diffusion initiated by $S$ can be described as a discrete-time stochastic process under the IC-model \cite{kempe2003maximizing}. Let $S_i\subseteq V$ be the active node set at time step $t_i$. At time step $t_0$, all nodes in $S$ are active, namely $S_0:=S$. We call $S_0$ the seed set, and node in this set is the seed of this cascade. At time step $t_i$, $i\geq 1$, we set $S_i:=S_{i-1}$ first; then, for those nodes activated first at time step $t_{i-1}$, $u\in(S_{t-1}\backslash S_{t-2})$, it activates its each inactive outgoing neighbor $v$ with the probability $p_{uv}$ by one chance. If $u$ activates $v$ at $t_i$ successfully, we add $v$ into $S_i$. The influence diffusion terminates when no more inactive nodes can be activated.

The above influence diffusion process can be interpreted by sampling a graph realization. Given a directed network $G=(V,E)$, we can decide whether an edge $(u,v)\in E$ is live or blocked with probability $p_{uv}$. To remove these blocked edges, the remaining graph is a subgraph $g$ of $G$. This subgraph $g$ is called ``graph realization''. These edges existed in $g$ are known as live edges, or else called blocked edges. For each edge $(u,v)\in E$, it exists in a graph realization $g$ with probability $p_{uv}$ under the IC-model. There are $2^m$ possible graph realizations altogether under the IC-model. Let $\mathcal{G}$ be the set of all possible graph realizations with $|\mathcal{G}|=2^m$, and $g$ be a graph realization sampled from $\mathcal{G}$, denoted by $g\leftarrow\mathcal{G}$, with probability as follows:
\begin{equation}
\Pr[g|g\leftarrow\mathcal{G}]=\prod_{e\in E(g)}p_e\prod_{e\in E(G)\backslash E(g)}(1-p_e)
\end{equation}
\begin{rem}
	In most references, they usually called ``graph realization'' as ``realization'' or ``possible world''. They all refer to an instance of a probabilistic social network. We will discuss a different concept ``realization'' later. To avoid ambiguity, we use ``graph realization'' here.
\end{rem}
The problem and algorithms discussed later in this paper are based on the IC-model, because the adaptive IM problem can be adaptive submodular only under the IC-model.

\subsection{Adaptive Influence Maximization}
Given a seed set $S\subseteq V$ and a graph realization $g\in\mathcal{G}$, the size of final active set that can be reached by the seed set $S$ under the graph realization $g$ is denoted by $\sigma_g(S)$. Thus, the expected influence spread $\sigma_G(S)$ under the IC-model can be defined as follows:
\begin{equation}\label{eq2}
\sigma_G(S)=\mathbb{E}_{g\leftarrow\mathcal{G}}[\sigma_g(S)]=\sum_{g\in\mathcal{G}}\sigma_g(S)\cdot\Pr[g|g\leftarrow\mathcal{G}]
\end{equation}
where it is the weighted expectation of influence spread over all possible graph realizations. The influence maximization (IM) problem aims to find a seed set $S\subseteq V$, such that $|S|\leq k$, to maximize the expected influence spread $\sigma_G(S)$.

From the above, in this non-adaptive setting, the seed set $S$ is selected once without the knowldge of what graph realization happens in the actual diffusion process. Thus, the actual influence spread of $S$ may be much worse than our expectation. Instead, in an adaptive manner, we select a node $u$ from $V$ at a time and wait to observe the actual diffusion result. Relied on this observation, we select the next node that could activate those inactive nodes as much as possible. It is called full-adoption feedback model \cite{golovin2011adaptive}. In other words, when we select a node $u$ as seed, we are able to know the status of all edges going out from those nodes that can be reached by $u$ through live edges in current graph realization. Golovin \textit{et al.} \cite{golovin2011adaptive} introduced two important concepts, adaptive monotonicity and adaptive submodularity, and showed that the simple adaptive greedy policy has a $(1-1/e)$-approximation guarantee.

\subsection{Problem Definition}
In the traditional IM problem, it assumes that each user in the seed set we select is activated successfully and then spread our given information cascade. However, in the real scenario, not all users in the seed set are willing to be an influencer. Based on that, we consider a user can be activated as a seed with a certain probability and we can attempt to activate her many times. For each user $u\in V$, there is a probability $\beta_u\in(0,1]$ with which she can be activated successfully when we select her as a seed. Let $\mathbb{Z}^V_+$ be the collection of non-negative integer vector, each component is indexed by a node in $V$. For a vector $\vec{x}\in\mathbb{Z}^V_+$, if $\vec{x}(u)=i$, it means that we select user $u$ and try to activate her as a seed $i$ times.

Given a social graph $G=(V,E)$, we have a total budget $k\in\mathbb{R}_+$, a vector $\vec{b}\in\mathbb{Z}^V_+$, and a cost function $c:V\times \mathbb{Z}_+\rightarrow\mathbb{R}_+$. Here, for each user $u\in V$, we assume she can be tried to activate as a seed at most $\vec{b}(u)$ times, and it costs $c(\langle u,i\rangle)$ when the $i$-th trial of activating user $u$ as a seed happens. Given a seeding vector $\vec{x}$ and a graph realization $g\in\mathcal{G}$, the expected number of active nodes $\mu_g(\vec{x})$ under the graph realization $g$ can be defined as follows:
\begin{flalign}
\mu_g(\vec{x})&=\mathbb{E}_{S\leftarrow\vec{x}}[\sigma_g(S)]=\sum_{S\subseteq V}\sigma_g(S)\cdot\Pr[S|S\leftarrow\vec{x}]\\
&=\sum_{S\subseteq V}\sigma_g(S)\cdot\prod_{u\in S}\left(1-(1-\beta_u)^{\vec{x}(u)}\right)\cdot\prod_{u\in V\backslash S}(1-\beta_u)^{\vec{x}(u)}
\end{flalign}
Similar to Equation (\ref{eq2}), we have
\begin{equation}
\mu_G(\vec{x})=\mathbb{E}_{g\leftarrow\mathcal{G}}[\mu_g(\vec{x})]=\sum_{g\in\mathcal{G}}\mu_g(x)\cdot\Pr[g|g\leftarrow\mathcal{G}]
\end{equation}
where $\mu_G(\vec{x})$ is the expected influence spread over all possible graph realizations given a seeding vector $\vec{x}$. The IM on multiple activations (IMMA) problem is formulated, which seeks a seeding vector $\vec{x}\in\mathbb{Z}^V_+$ that maximizes $\mu_G(\vec{x})$ subject to $c(\vec{x})\leq k$ and $\vec{x}\leq\vec{b}$. Here, we denote $c(\vec{x})$ by $c(\vec{x})=\sum_{u\in V}\sum_{i\in[\vec{x}(u)]}c(\langle u,i\rangle)$, where $[j]=\{1,2,\cdots,j\}$. Each trial is independent.

In the adaptive setting, the IMMA problem can be transformed to find a policy $\pi$, where we select seed nodes step by step. The parameter setting is the same as before. A seeding vector is initialized to $\vec{x}=\vec{0}\in \mathbb{Z}^V_+$. When selecting an inactive user $u\in V$ with $\vec{x}(u)<\vec{b}(u)$, we increase $\vec{x}(u)$ by $1$ and attempt to activate $u$ to be an active seed with probability $\beta_u$. At this moment, we need to observe two states as follows: (1) Node state: whether user $u$ can be activated to be an active seed successfully; and (2) Edge state: If $u$ becomes an active seed, wait to observe the influence diffusion process (related edges is live or blocked) until no new nodes can be activated. We repeated this process until no remaining budget exists.

Next, we define the states of the given network. Given a social graph $G=(V,E)$, for each node $u\in V$, the state of $u$ can be denoted by $X_u\in\{0,1,?\}^{\vec{b}(u)}$, where $X_u(i)=1$ means user $u$ is activated as a seed successfully in the $i$-th trial, or not succeed, $X_u(i)=0$. $X_u(i)=?$ if the result of $i$-th trial is unknown. Similar, for each edge $(u,v)\in E$, the state of $(u,v)$ can be denoted by $Y_{uv}\in\{0,1,?\}$, where $Y_{uv}=1$ means edge $(u,v)$ is live, and $Y_{uv}=0$ means edge $(u,v)$ is blocked. $Y_{uv}=?$ if the state of $(u,v)$ is unknown. At the beginning, the states of all nodes and edges are $?$. After defining the state variables, we have a function $\phi$ mapping like
\begin{equation}
\phi:\{X_u\}_{u\in V}\cup\{Y_{uv}\}_{(u,v)\in E}\rightarrow\left\{\{0,1\}^{\vec{b}(u)}\right\}_{u\in V}\cup\{0,1\}^E
\end{equation}
where $\phi$ is called a realization (full realization), where the states of all items are known. We say that $\phi(u)\in\{0,1\}^{\vec{b}(u)}$ is the state of user $u\in V$, $\phi(u)(i)\in\{0,1\}$ is the state of the $i$-th trial for user $u$, and $\phi((u,v))\in\{0,1\}$ is the state of edge $(u,v)\in E$ under the realization $\phi$. Let $\Phi$ be the set of all possible realizations. We define $\phi$ as a realization sampled from $\Phi$, denoted by $\phi\leftarrow\Phi$, with probability $\Pr[\phi|\phi\leftarrow\Phi]$. That is
\begin{equation}
	\Pr[\phi|\phi\leftarrow\Phi]=\prod_{e\in E\atop\phi(e)=1}p_e\prod_{e\in E\atop\phi(e)=0}(1-p_e)\cdot\prod_{u\in V}\left[\prod_{i\in[\vec{b}(u)]\atop\phi(u)(i)=1}\beta_u\prod_{i\in[\vec{b}(u)]\atop\phi(u)(i)=0}(1-\beta_u)\right]
\end{equation}

In this adaptive seeding process, after the $i$-th trial to activate node $u$ is finished, its state and the states of those related edges could be updated. Our observation until now can be described by a partial realization $\psi$. It is a function of observed items to their states. For $u\in V$, $\psi(u)\in\{0,1,?\}^{\vec{b}(u)}$, and $\psi(u)(i)=?$ if the result $i$-th trial to node $u$ is not yet observed. For $(u,v)\in E$, $\psi((u,v))\in\{0,1,?\}$ as well. The domain of a partial realization $\psi$ can be defined as $\text{dom}(\psi)=\{\langle u,i\rangle | \psi(u)(i)\neq ?\}$, which is the trials that have been done. We say $\psi$ is consistent with a realization $\phi$ if the states of items in the domain of $\psi$ are equal between them, denoted by $\phi\sim\psi$. Given $\psi$, $\psi'$ and $\phi$, if $\text{dom}(\psi)\subseteq\text{dom}(\psi')$ and $\phi\sim\psi,\psi'$, we say $\psi$ is a subrealization of $\psi'$, denoted by $\psi\subseteq\psi'$.

Let $\pi(\kappa)$ be a randomized policy based on a random variable $\kappa$ that represents a random source of this randomized policy. The $\pi(\kappa)$ is a function mapping from current seeding vector $\vec{x}$ and one of its possible partial realizations $\psi$ to a node $u^*$, then it executes the $(\vec{x}(u^*)+1)$-th trial that tries to select node $u^*$ as a seed. Here, we denote by $u^*=\pi(\kappa,\vec{x},\psi)$ where $u^*$ is the next potential seed that policy $\pi(\kappa)$ will select based on $\vec{x}$ and $\psi$. The influence spread gained from policy $\pi(\kappa)$ under the realization $\phi$ can be defined as follows:
\begin{equation}\label{eq8}
f(\eta(\pi(\kappa),\phi),\phi)=\sigma_{g_\phi}\left(\{u|\exists_{1\leq j\leq\eta(\pi(\kappa),\phi)(u)}\phi(u)(j)=1\}\right)
\end{equation}
where $g_\phi\in\{0,1\}^E$ is the graph realization $g$ contained in realization $\phi$ and $\eta(\pi(\kappa),\phi)$ is the seeding vector returned by policy $\pi$ under the realization $\phi$. The expected influence spread of policy $\pi(\kappa)$ can be shown as follows:
\begin{equation}\label{eq9}
\mathbb{E}_\kappa\left[f_{avg}(\pi(\kappa))\right]=\mathbb{E}_\kappa\left[\mathbb{E}_{\phi\leftarrow\Phi}\left[f(\eta(\pi(\kappa),\phi),\phi)\right]\right]
\end{equation}

Therefore, the adaptive IM on multiple activations (Adaptive-IMMA) problem is formulated, which can be defined in Problem \ref{pro1}.
\begin{pro}[Adaptive-IMMA Problem]\label{pro1}
	Givne a social graph $G=(V,E)$, a budget $k\in\mathbb{R}_+$, a vector $\vec{b}\in\mathbb{Z}^V_+$, and a cost function $c:V\times \mathbb{Z}_+\rightarrow\mathbb{R}_+$, it aims to find a policy $\pi^*(\kappa)$ that maximizes its expected influence spread defined in Equation (\ref{eq9}), i.e., $\pi^*\in\arg\max_{\pi}\mathbb{E}_\kappa[f_{avg}(\pi(\kappa))]$ subject to $\eta(\pi(\kappa),\phi)\leq\vec{b}$ and $\mathbb{E}_\kappa[c(\eta(\pi(\kappa),\phi))]\leq k$ for all realizations $\phi$.
\end{pro}
\noindent
Given a seed vector $\vec{x}\in\mathbb{Z}_+^V$, we say ``increase $\vec{x}(u)$ by $1$'' is equivalent to execute the trial $\langle u,\vec{x}(u)+1\rangle$. For each node $u\in V$, we assume $c(\langle u,1\rangle)\leq c(\langle u,2\rangle)\leq\cdots\leq c(\langle u,\vec{b}(u)\rangle)$. It is valid becuase in general, we execute trial $\langle u,i+1\rangle$ only when trial $\langle u,i\rangle$ fails to activate node $u$ as a seed, thereby the cost of trial $\langle u,i+1\rangle$ is larger than the cost of $\langle u,i\rangle$.

\section{The Properties}
In this section, we first introduce some concepts of submodularity on integer lattice, and then, generalize several properties of our Adaptive-IMMA problem.

\subsection{Submodular function on integer lattice}
Usually, for two sets $S,T\subseteq V$, a set function $h:2^V\rightarrow\mathbb{R}_+$ is monotone if $h(S)\leq f(T)$ for any $S\subseteq T\subseteq V$ and submodular if $h(S)+h(T)\geq h(S\cup T)+h(S\cap T)$. The submodularity of set function can be generalized by diminishing return property, in other words, submodular if $h(S\cup\{u\})-f(S)\geq f(T\cup\{u\})-f(T)$ for any $S\subseteq T\subseteq V$ and $u\notin T$. On integer lattice, for two vectors $\vec{s},\vec{t}\in\mathbb{Z}^V_+$, let $\vec{s}\lor\vec{t}\in\mathbb{Z}^V_+$ be defined as $(\vec{s}\lor\vec{t})(u)=\max\{\vec{s}(u),\vec{t}(u)\}$, and $\vec{s}\land\vec{t}\in\mathbb{Z}^V_+$ be defined as $(\vec{s}\land\vec{t})(u)=\min\{\vec{s}(u),\vec{t}(u)\}$ for any $u\in V$. A vector function $f:\mathbb{Z}^V_+\rightarrow\mathbb{R}_+$ is defined on the integer lattice $\mathbb{Z}^V_+$. This vector function $f$ is monotone if $f(\vec{s})\leq f(\vec{t})$ for any $\vec{s}\leq\vec{t}\in\mathbb{Z}^V_+$ and lattice submodular if $f(\vec{s})+f(\vec{t})\geq f(\vec{s}\lor\vec{t})+f(\vec{s}\land\vec{t})$ for any $\vec{s},\vec{t}\in\mathbb{Z}^V_+$. When the domain of vector are restricted to binary lattice $\{0,1\}^V$, the vector function $f$ is reduced to set function $h$. Thus, the submodularity on set function is a special case of submodularity on integer lattice.

Besides, we consider a vector function $f:\mathbb{Z}^V_+\rightarrow\mathbb{R}_+$ is diminishing return submodular (dr-submodular) if $f(\vec{s}+\vec{e}_u)-f(\vec{s})\geq f(\vec{t}+\vec{e}_u)-f(\vec{t})$ for any $\vec{s}\leq\vec{t}$ and $u\in V$, where $\vec{e}_u\in\mathbb{Z}^V$ is the $u$-th unit vector with the $u$-th component being $1$ and others being $0$. Here, there is a little different from the submodularity on set function. $f$ is lattice submodular does not mean it is dr-submodular on integer lattice, but the opppsite is true. Thus, dr-submodularity is a stronger property than lattice submodular. We consider the dr-submodularity later.

\subsection{Properties of the Adaptive-IMMA}
Assume that $\beta_u=1$ for each node $u\in V$ and seeding vector $\vec{x}\in\{0,1\}^V$, the IMMA problem can be reduced to the IM problem naturally. Therefore, the IMMA problem is more general and inherits the NP-hardness of IM. In the traditional IM problem, the expected influence spread shown as Equation (\ref{eq2}) is monotone and submodular on the seed set \cite{kempe2003maximizing}. In order to study the properties of our Adaptive-IMMA problem, we define its marginal gain first, that is
\begin{defn}[Conditional Expected Marginal Gain on Integer Lattice]\label{def1}
	Given a seeding vector $\vec{x}\in\mathbb{Z}^V_+$ and a partial realization $\psi$ generated by it, the conditional expected marginal gain of increasing $\vec{x}(u)$ by $1$ is defined as
	\begin{equation}\label{eq10}
	\Delta(u|\vec{x},\psi)=\mathbb{E}_{\phi\sim\psi}[f(\vec{x}+\vec{e}_u,\phi)-f(\vec{x},\phi)]
	\end{equation}
	where the expectation is on $\Pr[\phi|\phi\sim\psi]$. The condidional expected marginal gain of policy $\pi(\kappa)$ is defined as
	\begin{equation}
	\Delta(\pi(\kappa)|\vec{x},\psi)=\mathbb{E}_{\phi\sim\psi}[f(\vec{x}\lor\eta(\pi(\kappa),\phi),\phi)-f(\vec{x},\phi)]
	\end{equation}
\end{defn}
\noindent
Here, $\Delta(u|\vec{x},\psi)$ is the expected gain by increasing $\vec{x}(u)$ by $1$ conditioned on current partial realization $\psi$ of $\vec{x}$ and $\Delta(\pi(\kappa)|\vec{x},\psi)$ is the expected gain by running $\pi(\kappa)$ after observing partial realization $\psi$ but neglect it. Then, adapted from \cite{golovin2011adaptive}, the concepts of adaptive monotonicity and adaptive submodularity are described as follows:
\begin{defn}[Adaptive Monotonicity]
	A vector function $f(\cdot,\phi)$ is adaptive monotone if the conditional expected marginal gain with resprect to distribution $\Pr[\phi]$ of any node $u$, seeding vector $\vec{x}$, and its possible partial realization $\psi$ is nonnegative, that is
	\begin{equation}
	\Delta(u|\vec{x},\psi)\geq 0
	\end{equation}
\end{defn}
\begin{defn}[Adaptive dr-submodularity]
	A vector function $f(\cdot,\phi)$ is adaptive dr-submodular if the conditional expected marginal gain with resprect to distribution $\Pr[\phi]$ of any node $u$, seeding vectors $\vec{x},\vec{y}$ with $\vec{x}\leq\vec{y}$, and their possible partial realizations $\psi$ (generated by $\vec{x}$), $\psi'$ (generated by $\psi'$) with $\psi\subseteq\psi'$ satisfies the following inequality, that is
	\begin{equation}
	\Delta(u|\vec{x},\psi)\geq\Delta(u|\vec{y},\psi')
	\end{equation}
\end{defn}
\noindent
For our Adaptive-IMMA problem, the function $f(\cdot,\phi)$ is adaptive monotone and adaptive submodular according to Theorem \ref{thm1} and Theorem \ref{thm2}.
\begin{thm}\label{thm1}
	The objective function $f(\cdot,\phi)$ is adaptive monotone.
\end{thm}
\begin{proof}
	To prove adaptive monotonicity, we are required to show $\Delta(u|\vec{x},\psi)\geq 0$. Given a seeding vector $\vec{x}$ and its partial realization $\psi$, we denote the marginal gain under the realization $\phi\sim\psi$ as follows:
	\begin{equation}
	\Delta(u|\vec{x},\phi\sim\psi)=f(\vec{x}+e_u,\phi)-f(\vec{x},\phi)
	\end{equation}
	If node $u$ has been activated as a seed under the partial realization $\psi$, there is no marginal gain. Otherwise, it is possible to be activated by increasing $\vec{x}(u)$ by $1$, namely trial $\langle u,\vec{x}(u)+1\rangle$ succeeds. Thus, we have $\Delta(u|\vec{x},\phi\sim\psi)\geq 0$. The conditional expected marginal gain $\Delta(u|\vec{x},\psi)$ is a linear combination of all realizations $\phi\sim\psi$, thereby we have $\Delta(u|\vec{x},\psi)\geq 0$.
\end{proof}
\begin{thm}\label{thm2}
	The objective function $f(\cdot,\phi)$ is adaptive dr-submodular.
\end{thm}
\begin{proof}
	To prove its adaptive dr-submodularity, we are required to show $\Delta(u|\vec{x},\psi)\geq\Delta(u|\vec{y},\psi')$ for any two partial realizations $\psi$, $\psi'$ such that $\vec{x}\leq\vec{y}$ and $\psi\subseteq\psi'$. Considering two fixed partial realizations with $\psi\subseteq\psi'$, which are generated by seeding vector $\vec{x}$ and $\vec{y}$ respectively. We defined the generative active seed set $S$ under the seeding vector $\vec{x}$ and its partial realization $\psi$ as
	\begin{equation}
	S(\vec{x},\phi\sim\psi)=\left\{u\in V|\exists_{1\leq j\leq\vec{x}(u)}\phi(u)(j)=1\right\}
	\end{equation}
	Obviously, we have $S(\vec{x},\phi\sim\psi)\subseteq S(\vec{y},\phi'\sim\psi')$ because of $\vec{x}\leq\vec{y}$ and $\psi\subseteq\psi'$. Here, we assume two fixed realizaitons, $\phi\sim\psi$, $\phi'\sim\psi'$, and $d(u)=\vec{y}(u)-\vec{x}(u)$. For each $\langle u,i\rangle\notin\text{dom}(\psi')$, we have $\phi(u)(i-d(u))=\phi'(u)(i)$; for each $(u,v)\notin g(\psi')$, we have $\phi((u,v))=\phi'((u,v))$. We define the area that these two fixed realizations $\phi$, $\phi'$ share as $\alpha$. To show $\Delta(u|\vec{x},\phi\sim\psi)\geq \Delta(u|\vec{y},\phi'\sim\psi')$, we can consider these three cases:
	\begin{enumerate}
		\item $u\in S(\vec{x},\phi\sim\psi)$: We have $S(\vec{x},\phi\sim\psi)=S(\vec{x}+\vec{e}_u,\phi\sim\psi)$ and $S(\vec{y},\phi'\sim\psi')=S(\vec{y}+\vec{e}_u,\phi'\sim\psi')$. Thus, $\Delta(u|\vec{x},\phi\sim\psi)=\Delta(u|\vec{y},\phi'\sim\psi')$.
		\item $u\in S(\vec{y},\phi'\sim\psi')\backslash S(\vec{x},\phi\sim\psi)$: We have $S(\vec{x},\phi\sim\psi)\subseteq S(\vec{x}+\vec{e}_u,\phi\sim\psi)$ and $S(\vec{y},\phi'\sim\psi')=S(\vec{y}+\vec{e}_u,\phi'\sim\psi')$. Thus, $\Delta(u|\vec{x},\phi\sim\psi)\geq\Delta(u|\vec{y},\phi'\sim\psi')$.
		\item $u\in V\backslash S(\vec{y},\phi'\sim\psi')$: When $\vec{x}(u)=\vec{y}(u)=i$, we have $S(\vec{x}+\vec{e}_u,\phi\sim\psi)=S(\vec{x},\phi\sim\psi)\cup\{u\}$ and $S(\vec{y}+\vec{e}_u,\phi'\sim\psi')=S(\vec{y},\phi'\sim\psi')\cup\{u\}$ if $\phi(u)(i+1)=\phi'(u)(i+1)=1$. It inherits the adaptive submodularity of the adaptive IM problem under the full-adoption feedback model \cite{golovin2011adaptive}, thereby we have $\Delta(u|\vec{x},\phi\sim\psi)\geq\Delta(u|\vec{y},\phi'\sim\psi')$; or else there is no marginal gain if $\phi(u)(i+1)=\phi'(u)(i+1)=0$. When $\vec{x}(u)=i<j=\vec{y}(u)=$, we have $\Delta(u|\vec{x},\phi\sim\psi)\geq\Delta(u|\vec{y},\phi'\sim\psi')$ apparently as well if $\phi(u)(i+1)=\phi'(j+1)=1$; or else there is no marginal gain if $\phi(u)(i+1)=\phi'(j+1)=0$.
	\end{enumerate}
 	
 	From the above, we have known $\Delta(u|\vec{x},\phi\sim\psi)\geq \Delta(u|\vec{y},\phi'\sim\psi')$. According to Eqaution (\ref{eq10}) and Definition \ref{def1}, we have as follows:
	\begin{flalign}
	\Delta(u|\vec{x},\psi)&=\sum_{\phi\sim\psi}\Pr[\phi|\phi\sim\psi]\Delta(u|\vec{x},\phi\sim\psi)\\
	&=\sum_{\phi'\sim\psi'}\Pr[\phi'|\phi'\sim\psi']\sum_{\phi\sim\alpha}\Pr[\phi|\phi\sim\alpha]\Delta(u|\vec{x},\phi\sim\psi)\label{eq17}
	\end{flalign}
	Since $\sum_{\phi\sim\alpha}\Pr[\phi|\phi\sim\alpha]=1$, we have
	\begin{flalign}
	(\ref{eq17})&\geq\sum_{\phi'\sim\psi'}\Pr[\phi'|\phi'\sim\psi']\sum_{\phi\sim\alpha}\Pr[\phi|\phi\sim\alpha]\Delta(u|\vec{y},\phi'\sim\psi')\\
	&\geq\sum_{\phi'\sim\psi'}\Pr[\phi'|\phi'\sim\psi']\Delta(u|\vec{y},\phi'\sim\psi')\sum_{\phi\sim\alpha}\Pr[\phi|\phi\sim\alpha]\\
	&=\sum_{\phi'\sim\psi'}\Pr[\phi'|\phi'\sim\psi']\Delta(u|\vec{y},\phi'\sim\psi')\\
	&=\Delta(u|\vec{y},\psi')
	\end{flalign}
	Therefore, we have $\Delta(u|\vec{x},\psi)\geq\Delta(u|\vec{y},\psi')$ for any $\vec{x}\leq\vec{y}$ and their partial realizations $\psi\subseteq\psi'$. The proof of adaptive submodular is completed.
\end{proof}

\section{Algorithm and theoretical Analysis}
In this section, we propose algorithms to solve our Adaptive-IMMA problem and give an approximation ratio with necessary theoretical analysis.

\subsection{Adaptive Greedy Policy}
We define a randomized adaptive greedy policy $\pi^g(\kappa)$ here. The seeding vector $\vec{x}$ is initialized to $\vec{x}=\vec{0}\in\mathbb{Z}^V_+$. In each iteration, the $\pi^g(\kappa)$ selects the node $u^*\in V$ that maximizes $\Delta(u|\vec{x},\psi)/c(\langle u,\vec{x}(u)+1\rangle)$ where $\vec{x}(u)<\vec{b}(u)$ and $\psi$ is the partial realization generated by the current $\vec{x}$, then increases $\vec{x}(u^*)$ by $1$. Then, we need to observe the state of $u^*$ and update this partial realization $\psi$. The $\pi^g(\kappa)$ repeats above procedure, terminates until $c(\vec{x})\geq k$, or terminates with a probability. The main idea of adaptive greedy policy is shown in Algorithm \ref{a1}. Shown as line 5 to 7 of Algorithm \ref{a1}, it returns with a probability when the remaining budget is not sufficient to do a trial on the selected node $u^*$. Thereby, the random source $\kappa$ in this adaptive greedy policy indicates whether contains the selected node in the last iteration. The adaptive greedy policy $\pi^g(\kappa)$ shown as Algorithm \ref{a1} satisfies $\eta(\pi(\kappa),\phi)\leq\vec{b}$ and $\mathbb{E}_\kappa[c(\eta(\pi(\kappa),\phi))]\leq k$ for any realization $\phi$.

\begin{algorithm}[!t]
	\caption{\text{AdaptiveGreedy $(G,f,k,\vec{b},c)$}}\label{a1}
	\begin{algorithmic}[1]
		\renewcommand{\algorithmicrequire}{\textbf{Input:}}
		\renewcommand{\algorithmicensure}{\textbf{Output:}}
		\REQUIRE A graph $G=(V,E)$, a function $f(\cdot,\phi)$, a budget $k\in\mathbb{R}_+$, a vector $\vec{b}\in\mathbb{Z}^V_+$ and, a cost function $c:V\times\mathbb{Z}_+\rightarrow\mathbb{R}_+$
		\ENSURE A seeding vector $\vec{x}\in\mathbb{Z}^V_+$ and $f(\vec{x},\psi)$
		\STATE Initialize: $\vec{x}:=\vec{0}$
		\STATE Initialize: $\psi:=\left\{\{?\}^{\vec{b}(u)}\right\}_{u\in V}\cup\{?\}^E$
		\WHILE {$c(\vec{x})<k$}
		\STATE $u^*\in\arg\max_{u\in V,\vec{x}(u)<\vec{b}(u)}\Delta(u|\vec{x},\psi)/c(\langle u,\vec{x}(u)+1\rangle)$
		\IF {$c(\vec{x})+c(\langle u^*,\vec{x}(u^*)+1\rangle)>k$}
		\STATE $\textbf{break}$ with probability $1-(k-c(\vec{x}))/c(\langle u^*,\vec{x}(u^*)+1\rangle)$
		\ENDIF
		\STATE $\vec{x}(u^*):=\vec{x}(u^*)+1$
		\STATE Observe the state of $\langle u^*,\vec{x}(u^*)\rangle$
		\STATE Update $\psi:=\psi\cup \langle u^*,\vec{x}(u^*)\rangle$
		\IF {$\psi(u^*)(\vec{x}(u^*)) = 1$}
		\STATE Update the edge states of $\psi$ observed by $u^*$'s actual influence diffusion 
		\ENDIF
		\ENDWHILE
		\RETURN $\vec{x}$, $f(\vec{x},\psi)$
	\end{algorithmic}
\end{algorithm}

\subsection{Theoretical Analysis}
To make the following analysis understandable, we introduce the operations of policy truncation and policy concatenation, which are adapted from \cite{golovin2011adaptive} but suitable on integer lattice domain. We imagine a randomized policy $\pi(\kappa)$ running over time. In each iteration, it selects node $u^*=\pi(\kappa,\vec{x},\psi)$ under the current seeding vector $\vec{x}$ and its partial realization $\psi$. It runs trial $\langle u^*,\vec{x}(u^*)+1\rangle$ for $c(\langle u^*,\vec{x}(u^*)+1\rangle)$ units of time and increases $\vec{x}(u^*)$ by $1$.
\begin{defn}[Policy Truncation]
	Let the seeding vector $\vec{x}$ kept by $\pi(\kappa)$, the policy truncation $\pi_{[t]}(\kappa)$ denotes the randomized policy that runs $\pi(\kappa)$ for $t$ units of time. If the last trial $\langle u^*,\vec{x}(u^*)+1\rangle$ can only be run for $0\leq\tau<c(\langle u^*,\vec{x}(u^*)+1\rangle)$ time, it will increase $\vec{x}(u^*)$ by $1$ with probability $\tau\backslash c(\langle u^*,\vec{x}(u^*)+1\rangle)$. Under any realization $\phi$, we have $\mathbb{E}_\kappa[c(\eta(\pi_{[t]}(\kappa),\phi))]\leq t$.
\end{defn}

\begin{defn}[Policy Concatenation]
	For any two adaptive policies $\pi(\kappa)$ and $\pi'(\kappa)$, the policy concatenation $\pi(\kappa)@\pi'(\kappa)$ denotes the adaptive policy that runs policy $\pi(\kappa)$ first, and then runs $\pi'(\kappa)$ like a fresh start without information from the run of $\pi(\kappa)$. Under any realization $\phi$, we have $\eta(\pi(\kappa)@\pi'(\kappa),\phi)=\eta(\pi(\kappa),\phi)\lor\eta(\pi'(\kappa),\phi)$.
\end{defn}

\begin{lem}\label{lem1}
	The objective function $f(\cdot,\phi)$ is adaptive monotone if and only if for any randomized policies $\pi(\kappa)$ and $\pi'(\kappa)$, we have
	\begin{equation}\label{eq22}
		\mathbb{E}_\kappa\left[f_{avg}(\pi(\kappa))\right]\leq\mathbb{E}_\kappa\left[f_{avg}(\pi'(\kappa)@\pi(\kappa))\right]
	\end{equation}
\end{lem}
\begin{proof}
	Given a fixed random source $\kappa$, we have $\eta(\pi'(\kappa)@\pi(\kappa),\phi)=\eta(\pi'(\kappa),\phi)\lor\eta(\pi(\kappa),\phi)=\eta(\pi(\kappa)@\pi'(\kappa),\phi)$ under any realization $\phi$. Therefore, $f_{avg}(\pi(\kappa))\leq f_{avg}(\pi'(\kappa)@\pi(\kappa))$ holds if and only if $f_{avg}(\pi(\kappa))\leq f_{avg}(\pi(\kappa)@\pi'(\kappa))$. Then, we need to show $f_{avg}(\pi(\kappa))\leq f_{avg}(\pi(\kappa)@\pi'(\kappa))$, which can be inferred from Lemma A.8 in \cite{golovin2011adaptive}, thus we omit here. Take the expectation over random source $\kappa$, Inequality (\ref{eq22}) can be established.
\end{proof}
\begin{lem}\label{lem2}
	Given a seeding vector $\vec{x}$ and a partial realization $\psi$ generated by it, $f(\cdot,\phi)$ is an adaptive monotone and adaptive dr-submodular function. For any policy $\pi^*(\kappa)$ that satisfies $\eta(\pi^*(\kappa),\phi)\leq\vec{b}$ and $\mathbb{E}_\kappa[c(\eta(\pi^*(\kappa),\phi))]\leq k$ for any realization $\phi$, we have
	\begin{equation}\label{eq23}
	\Delta(\pi^*(\kappa)|\vec{x},\psi)\leq\mathbb{E}_{\phi\sim\psi}\left[c(\eta(\pi^*(\kappa),\phi))\right]\cdot\max_{u\in V,\vec{x}(u)<\vec{b}(u)}\left\{\frac{\Delta(u|\vec{x},\psi)}{c(\langle u,\vec{x}(u)+1\rangle)}\right\}
	\end{equation}
\end{lem}
\begin{proof}
	Consider the seeding vector $\vec{x}'$ maintained by a policy $\pi'(\kappa)$, we can define this policy $\pi'(\kappa)$ as follow. The seeding vector $\vec{x}'$ is initialized to $\vec{x}'=\vec{0}\in\mathbb{Z}^V_+$. The policy $\pi'(\kappa)$ increases $\vec{x}'(u)$ from $0$ to $\vec{x}(u)$ step by step for each node $u\in V$. It will terminate if the state of trial $\langle u,i\rangle$ with $i\leq\vec{x}(u)$ is different from $\psi(u)(i)$ or the state of edge $(u,v)$ is different from $\psi((u,v))$. If reaching $\vec{x}'=\vec{x}$ and not stopping, it will begin to run policy $\pi^*(\kappa)$ like a fresh start without information from before. Here, we can imagine there is a virtual vector $\vec{x}^*$ associated with $\pi^*(\kappa)$ updated from $\vec{0}$ and $\vec{x}'=\vec{x}\lor\vec{x}^*$ under the realization $\phi\sim\psi$.
	
	For each trial $\langle u,i\rangle$, we define $w(\langle u,i\rangle)=\Pr[i\leq\eta (\pi'(\kappa),\phi)(u)|\phi\sim\psi]$ as the probability that $u$ is selected by $\pi'(\kappa)$ and increases $\vec{x}'(u)$ from $i-1$ to $i$. When the policy $\pi^*(\kappa)$ selects a node $u\in V$ with $\vec{x}(u)\leq\vec{x}^*(u)<\vec{b}(u)$, namely $\langle u,\vec{x}^*(u)+1\rangle\notin\text{dom}(\psi)$, the partial realization $\psi'$ generated by current $\vec{x}'$ satisfies $\psi\subseteq\psi'$, thereby we have $\Delta(u|\vec{x}',\psi')\leq\Delta(u|\vec{x},\psi)$ because of adaptive dr-submodularity. Thus, the total contribution to $\Delta(\pi^*(\kappa)|\vec{x},\psi)$ is bounded by $\Delta(\pi^*(\kappa)|\vec{x},\psi)\leq\sum_{u\in V,\vec{x}(u)<\vec{b}(u)}\sum_{i=\vec{x}(u)}^{\vec{b}(u)-1}w(\langle u,i+1\rangle)\cdot\Delta(u|\vec{x},\psi)$. From the above, we have
	\begin{flalign}
	\Delta(\pi^*(\kappa)|\vec{x},\psi)&\leq\sum_{u\in V,\vec{x}(u)<\vec{b}(u)}\sum_{i=\vec{x}(u)}^{\vec{b}(u)-1}w(\langle u,i+1\rangle)\cdot\Delta(u|\vec{x},\psi)\\
	&=\sum_{u\in V,\vec{x}(u)<\vec{b}(u)}\sum_{i=\vec{x}(u)}^{\vec{b}(u)-1}w(\langle u,i+1\rangle)\cdot c(\langle u,i+1\rangle)\cdot\frac{\Delta(u|\vec{x},\psi)}{c(\langle u,i+1\rangle)}\label{eq25}
	\end{flalign}
	Since $c(\langle u,i\rangle)\leq c(\langle u,i+1\rangle)$, we have
	\begin{flalign}
	(\ref{eq25})&\leq\sum_{u\in V,\vec{x}(u)<\vec{b}(u)}\frac{\Delta(u|\vec{x},\psi)}{c(\langle u,\vec{x}(u)+1\rangle)}\sum_{i=\vec{x}(u)}^{\vec{b}(u)-1}w(\langle u,i+1\rangle)\cdot c(\langle u,i+1\rangle)\\
	&\leq\left(\sum_{u\in V,\vec{x}(u)<\vec{b}(u)}\sum_{i=\vec{x}(u)}^{\vec{b}(u)-1}w(\langle u,i+1\rangle)\cdot c(\langle u,i+1\rangle)\right)\cdot\max_{u\in V,\vec{x}(u)<\vec{b}(u)}\left\{\frac{\Delta(u|\vec{x},\psi)}{c(\langle u,\vec{x}(u)+1\rangle)}\right\}\\
	&\leq\mathbb{E}_{\phi\sim\psi}[c(\eta(\pi^*(\kappa),\phi))]\cdot\max_{u\in V,\vec{x}(u)<\vec{b}(u)}\left\{\frac{\Delta(u|\vec{x},\psi)}{c(\langle u,\vec{x}(u)+1\rangle)}\right\}\label{eq28}
	\end{flalign}
	where Inequality (\ref{eq28}) is correct because it only count a subset of trials contained in $\eta(\pi^*(\kappa),\phi)$. Thus, this lemma is proven.
\end{proof}

\begin{thm}\label{thm3}
	The adaptive greedy policy $\pi^g(\kappa)$ shown as Algorithm \ref{a1} achieves a $(1-e^{-1})$ expected approximation guarantee. Thus, for any policy $\pi^*(\kappa)$ that satisfies $\eta(\pi^*(\kappa),\phi)\leq\vec{b}$ and $\mathbb{E}_\kappa[c(\eta(\pi^*(\kappa),\phi))]\leq k$ for any realization $\phi$, we have
	\begin{equation}
	\mathbb{E}_\kappa\left[f_{avg}(\pi^g(\kappa))\right]\geq\left(1-e^{-1}\right)\cdot\mathbb{E}_\kappa\left[f_{avg}(\pi^*(\kappa))\right]
	\end{equation}
\end{thm}
\begin{proof}
	Consider the policy $\pi^g_{[i+1]}(\kappa)$ given any $i\in[0,k-1]$, its current seeding vector and partial realization when it enters the last iteration before termination (line 3 of Algorithm \ref{a1}) are denoted by $\vec{x}$ and $\psi$. In the last iteration, the node $u^*$ is selected in line 4 of Algorithm \ref{a1}. The expected marginal gain of the last iteration is $\Delta(u^*|\vec{x},\psi)\cdot(i+1-c(\vec{x}))/c(\langle u^*,\vec{x}(u^*)+1\rangle)$. Consider the policy $\pi^g_{[i]}(\kappa)$, there are two cases could happen.
	\begin{enumerate}
		\item If $i\geq c(\vec{x})$, its execution will be the same as the policy $\pi^g_{[i+1]}(\kappa)$ until entering the last iteration. It updates $\vec{x}$ to $\vec{x}+\vec{e}_{u^*}$ with probability $(i-c(\vec{x}))/c(\langle u^*,\vec{x}(u^*)+1\rangle)$, which has $\Delta(u^*|\vec{x},\psi)\cdot(i-c(\vec{x}))/c(\langle u^*,\vec{x}(u^*)+1\rangle)$ expected marginal gain in the last iteration.
		\item If $i\leq c(\vec{x})$, the $\pi^g_{[i]}(\kappa)$ will not enter the last iteration of the policy $\pi^g_{[i+1]}(\kappa)$ obviously.
	\end{enumerate}
	
	According to the above analysis, the gap of the expected value of objective function returned by $\pi^g_{[i+1]}(\kappa)$ and $\pi^g_{[i]}(\kappa)$ can be bounded. We have
	\begin{equation}
		\mathbb{E}_\kappa\left[\mathbb{E}_{\phi\sim\psi}\left[f(\eta(\pi^g_{[i+1]}(\kappa),\phi),\phi)\right]\right]-\mathbb{E}_\kappa\left[\mathbb{E}_{\phi\sim\psi}\left[f(\eta(\pi^g_{[i]}(\kappa),\phi),\phi)\right]\right]\geq\frac{\Delta(u^*|\vec{x},\psi)}{c(\langle u^*,\vec{x}(u^*)+1\rangle)}
	\end{equation}
	Here, the $\vec{x}$ and $\psi$ are fixed, which are determined by potential realization $\phi$. Take the expectation over all realizations, we have
	\begin{equation}\label{eq31}
	\mathbb{E}_\kappa\left[f_{avg}(\pi^g_{[i+1]}(\kappa))\right]-\mathbb{E}_\kappa\left[f_{avg}(\pi^g_{[i]}(\kappa))\right]\geq\mathbb{E}_{\phi\leftarrow\Phi}\left[\frac{\Delta(u^*|\vec{x}_\phi,\psi_\phi)}{c(\langle u^*,\vec{x}_\phi(u^*)+1\rangle)}\right]
	\end{equation}
	where the $\vec{x}_\phi$ ($\psi_\phi$) is the current seeding vector (partial realization) of the policy $\pi^g_{[i+1]}(\kappa)$ at the beginning of its last iteration under the potential realization $\phi$ and the node $u^*$ can be considered as the one that is able to get the maximum marginal gain based on the seed vector $\vec{x}_\phi$ and its partial realizaton $\psi_\phi$.
	
	Then, the definition of $\vec{x}$ and $\psi$ are the same as above. We can define the seeding vector $\vec{y}$ and its partial realization $\psi'$ as that returned by the policy $\pi^g_{[i]}(\kappa)$, where we have $\psi\subseteq\psi'$ and $\vec{x}\leq\vec{y}$. Policy $\pi^g_{[i]}(\kappa)@\pi^*(\kappa)$ increase the value of objective function of policy $\pi^g_{[i]}(\kappa)$ by $\mathbb{E}_\kappa[\Delta(\pi^*(\kappa)|\vec{y},\psi')]$ expectedly. Besides, we have $\mathbb{E}_\kappa[\Delta(\pi^*(\kappa)|\vec{y},\psi')]\leq\mathbb{E}_\kappa[\Delta(\pi^*(\kappa)|\vec{x},\psi)]$ due to the adaptive dr-submodularity of $f(\cdot,\phi)$. Here, the $\vec{x}$ and $\psi$ are fixed, which are determined by potential realization $\phi$. Take the expectation over all realizations, we have
	\begin{equation}\label{eq32}
	\mathbb{E}_\kappa\left[f_{avg}(\pi^g_{[i]}(\kappa)@\pi^*(\kappa))\right]-\mathbb{E}_\kappa\left[f_{avg}(\pi^g_{[i]}(\kappa))\right]\leq\mathbb{E}_\kappa\left[\mathbb{E}_{\phi\leftarrow\Phi}\left[\Delta(\pi^*(\kappa)|\vec{x}_\phi,\psi_\phi)\right]\right]
	\end{equation}
	According to Inequality (\ref{eq23}) in Lemma \ref{lem2}, we have
	\begin{flalign}
		\mathbb{E}_\kappa\left[\Delta(\pi^*(\kappa)|\vec{x}_\phi,\psi_\phi)\right]&\leq\mathbb{E}_\kappa\left[\mathbb{E}_{\phi\sim\psi_\phi}\left[c(\eta(\pi^*(\kappa),\phi))\right]\right]\cdot\max_{u\in V,\vec{x}(u)<\vec{b}(u)}\left\{\frac{\Delta(u|\vec{x}_\phi,\psi_\phi)}{c(\langle u,\vec{x}_\phi(u)+1\rangle)}\right\}\\
		&\leq k\cdot\max_{u\in V,\vec{x}(u)<\vec{b}(u)}\left\{\frac{\Delta(u|\vec{x}_\phi,\psi_\phi)}{c(\langle u,\vec{x}_\phi(u)+1\rangle)}\right\}=k\cdot\frac{\Delta(u^*|\vec{x}_\phi,\psi_\phi)}{c(\langle u^*,\vec{x}_\phi(u^*)+1\rangle)}\label{eq34}
	\end{flalign}
	where Inequality (\ref{eq34}) is from $\mathbb{E}_\kappa[\mathbb{E}_{\phi\sim\psi_\phi}[c(\eta(\pi^*(\kappa),\phi))]]=\mathbb{E}_{\phi\sim\psi_\phi}[\mathbb{E}_\kappa[c(\eta(\pi^*(\kappa),\phi))]]\leq k$ since $\mathbb{E}_\kappa[c(\eta(\pi^*(\kappa),\phi))]\leq k$ for any realization $\phi$. Thus, we have
	\begin{flalign}
		(\ref{eq32})&=\mathbb{E}_{\phi\leftarrow\Phi}\left[\mathbb{E}_\kappa\left[\Delta(\pi^*(\kappa)|\vec{x}_\phi,\psi_\phi)\right]\right]\\
		&\leq k\cdot\mathbb{E}_{\phi\leftarrow\Phi}\left[\frac{\Delta(u^*|\vec{x}_\phi,\psi_\phi)}{c(\langle u^*,\vec{x}_\phi(u^*)+1\rangle)}\right]\leq k\cdot\left( \mathbb{E}_\kappa\left[f_{avg}(\pi^g_{[i+1]}(\kappa))\right]-\mathbb{E}_\kappa\left[f_{avg}(\pi^g_{[i]}(\kappa))\right]\right)\label{eq36}
	\end{flalign}
	
	Based on Lemma \ref{lem1}, we have $\mathbb{E}_\kappa[f_{avg}(\pi^*(\kappa))]\leq \mathbb{E}_\kappa[f_{avg}(\pi^g_{[i]}(\kappa)@\pi^*(\kappa))]$ becacuse of its adaptive monotonicity. According to Inequality (\ref{eq31}) (\ref{eq32}) (\ref{eq36}), we have
	\begin{equation}
	\mathbb{E}_\kappa\left[f_{avg}(\pi^*(\kappa))\right]-\mathbb{E}_\kappa\left[f_{avg}(\pi^g_{[i]}(\kappa))\right]\leq k\cdot\left(\mathbb{E}_\kappa\left[f_{avg}(\pi^g_{[i+1]}(\kappa))\right]-\mathbb{E}_\kappa\left[f_{avg}(\pi^g_{[i]}(\kappa))\right]\right)
	\end{equation}
	Now, we can define $\theta_i:=\mathbb{E}_\kappa[f_{avg}(\pi^*(\kappa))]-\mathbb{E}_\kappa[f_{avg}(\pi^g_{[i]}(\kappa))]$, which means $\theta_i\leq k\cdot(\theta_i-\theta_{i+1})$ and $\theta_{i+1}\leq(1-1/k)\cdot\theta_i$. Here, we have $\theta_k\leq(1-1/k)^k\cdot\theta_0\leq(1/e)\cdot\theta_0$, therefore $\mathbb{E}_\kappa[f_{avg}(\pi^*(\kappa))]-\mathbb{E}_\kappa[f_{avg}(\pi^g(\kappa))]\leq(e^{-1})\cdot(\mathbb{E}_\kappa[f_{avg}(\pi^*(\kappa))]-\mathbb{E}_\kappa[f_{avg}(\pi^g_{[0]}(\kappa))])$ when $k$ is relatively large. That is $\mathbb{E}_\kappa[f_{avg}(\pi^g(\kappa))]\geq(1-e^{-1})\cdot\mathbb{E}_\kappa[f_{avg}(\pi^*(\kappa))]$. The proof of this theorem is completed.
\end{proof}

\section{Solution Framework by Sampling}
In the last section, our adaptive greedy policy shown as Algorithm \ref{a1} can achieve a $(1-1/e)$ expected approximation guarantee, which has been proved by Theorem \ref{thm3}. However, it is based on a basic assumption that we are able to compute the exact value of $\Delta(u|\vec{x},\psi)$ and get the feasible node with maximum unit marginal gain in line 4 of Algorithm \ref{a1} in each iteration. In fact, this is an impossible task because it is \#P-hard \cite{chen2010scalable} to compute marginal gain $\Delta(u|\vec{x},\psi)$ for each node $u\in V$ under the IC-model. Thus, the true value of $\Delta(u|\vec{x},\psi)$ is difficult to obtain. MC simulations is a general method to estimate this value, but its running time is unacceptable. To overcome that, we are able to seek an estimator of $\Delta(u|\vec{x},\psi)$ through the reverse influence sampling (RIS) \cite{borgs2014maximizing} then maximize this estimator. If maximizing this estimator through sampling technique, it will be possible to get a extremely worse node with some probability, even though very small. In other words, the selected node $u^*$ in line 4 of Algorithm \ref{a1} is not optimal such that $u^*\notin\arg\max_{u\in V,\vec{x}(u)<\vec{b}(u)}\Delta(u|\vec{x},\psi)/c(\langle u,\vec{x}(u)+1\rangle)$ in actual execution. Like this, the expected approximation ratio shown in Theorem \ref{thm3} will not be ensured.

\subsection{Sampling Technique}
Consider the traditional IM problem, we need to introduce the concept of reverse reachable sets (RR-sets) first. Given a graph $G=(V,E)$, a random RR-set of $G$ can be generated by selecting a node $u\in V$ uniformly and sampling a graph realization $g$ from $\mathcal{G}$, then collecting those nodes can reach $u$ in $g$. A RR-set rooted at $u$ is a collection of nodes that are likely to influence $u$. A larger expected influence spread a seed set $S$ has, the higher the probability that $S$ intersects with a random RR-set is. Given a seed set $S$ and a random RR-set $R$, we have $\sigma_G(S)=n\cdot\Pr[R\cap S\neq\emptyset]$. Let $\mathcal{R}=\{R_1,R_2,\cdots,R_\theta\}$ be a collection of random RR-sets and $z(S,R)$ be the indicator, where $z(S,R)=1$ if $S\cap R\neq\emptyset$, or else $z(S,R)=0$. Denoted by $F_{\mathcal{R}}(S)=\sum_{i=1}^{\theta}z(S,R_i)/\theta$, the $n\cdot F_{\mathcal{R}}(S)$ is an unbiased estimator of the expected influence spread $\sigma_G(S)$. When the $|\mathcal{R}|$ is large, the $n\cdot F_{\mathcal{R}}(S)$ will converge to the true value $\sigma_G(S)$. Thus, how to set the value of $\theta$ is flexible, we need to balance between accuracy and running time carefully.

For our adaptive greedy policy, its current seeding vector and partial realization at the beginning of each iteration (when entering line 3 of Algorithm \ref{a1}) are denoted by $\vec{x}$ and $\psi$. Let $G(\psi)=(V(\psi),E(\psi))$ be the subgraph induced by all inactive nodes under the current partial realization $\psi$. Here, computing $\Delta(u|\vec{x},\psi)$ is equivalent to computing $\beta_u\cdot\sigma_{G(\psi)}(\{u\})$. We can note that $\Delta(u|\vec{x},\psi)=0$ if node $u\notin V(\psi)$. Thus, for a node $u\in V(\psi)$ and a random RR-sets $R(\psi)$ of $G(\psi)$, we can get an unbiased estimator of $\Delta(u|\vec{x},\psi)$. That is
\begin{flalign}
\Delta(u|\vec{x},\psi)&=\beta_u\cdot\sigma_{G(\psi)}(\{u\})\label{eq38}\\
&=\beta_u\cdot|V(\psi)|\cdot\Pr[\{u\}\cap R(\psi)\neq\emptyset]
\end{flalign}
Then, we can reformulate our adaptive greedy policy through the above sampling, which is shown in Algorithm \ref{a2}. It is called ``sampled adaptive greedy policy'' and denoted by $\pi^{gs}(\kappa,\omega)$, where the random variable $\omega$ indicates the random source of sampling for estimations. In each iteration, it generates a collection of random RR-sets $\mathcal{R}(\psi)=\{R_1(\psi),R_2(\psi),\cdots,R_\theta(\psi)\}$ based on current subgraph $G(\psi)$ first. Then, it select a feasible node $u^\circ\in V(\psi)$ that maximizes $\beta_u\cdot|V(\psi)|\cdot F_{\mathcal{R}(\psi)}(\{u\})/c(\langle u,\vec{x}(u)+1\rangle)$ where $\vec{x}(u)<\vec{b}(u)$ and increases $\vec{x}(u^\circ)$ by $1$. Finally, we need to observe the state of $u^\circ$, update this partial realization $\psi$, and update the subgraph $G(\psi)$. The $\pi^{gs}(\kappa,\omega)$ repeats above procedure, terminates until $c(\vec{x})\geq k$, or terminates with a probability.

\begin{algorithm}[!t]
	\caption{\text{Sampled-AdaptiveGreedy $(G,f,k,\vec{b},c,\varepsilon)$}}\label{a2}
	\begin{algorithmic}[1]
		\renewcommand{\algorithmicrequire}{\textbf{Input:}}
		\renewcommand{\algorithmicensure}{\textbf{Output:}}
		\REQUIRE A graph $G=(V,E)$, a function $f(\cdot,\phi)$, a budget $k\in\mathbb{R}_+$, a vector $\vec{b}\in\mathbb{Z}^V_+$, a cost function $c:V\times\mathbb{Z}_+\rightarrow\mathbb{R}_+$, and an error parameter $\varepsilon$
		\ENSURE A seeding vector $\vec{x}\in\mathbb{Z}^V_+$ and $f(\vec{x},\psi)$
		\STATE Initialize: $\vec{x}:=\vec{0}$
		\STATE Initialize: $\psi:=\left\{\{?\}^{\vec{b}(u)}\right\}_{u\in V}\cup\{?\}^E$
		\STATE Initialize: $G(\psi):=G$
		\STATE Initialize: $r$ be defined as Equation (\ref{eq41})
		\WHILE {$c(\vec{x})<k$}
		\STATE $u^\circ\leftarrow$ Generalized-EPIC $(G(\psi),\vec{x},\vec{b},c,\varepsilon)$
		\IF {$c(\vec{x})+c(\langle u^\circ,\vec{x}(u^\circ)+1\rangle)>k$}
		\STATE $\textbf{break}$ with probability $1-(k-c(\vec{x}))/c(\langle u^\circ,\vec{x}(u^\circ)+1\rangle)$
		\ENDIF
		\STATE $\vec{x}(u^\circ):=\vec{x}(u^\circ)+1$
		\STATE Observe the state of $\langle u^\circ,\vec{x}(u^\circ)\rangle$
		\STATE Update $\psi:=\psi\cup \langle u^\circ,\vec{x}(u^\circ)\rangle$
		\IF {$\psi(u^\circ)(\vec{x}(u^\circ)) = 1$}
		\STATE Update the edge states of $\psi$ observed by $u^\circ$'s actual influence diffusion
		\STATE Update $G(\psi)$ by removing all active nodes
		\ENDIF
		\ENDWHILE
		\RETURN $\vec{x}$, $f(\vec{x},\psi)$
	\end{algorithmic}
\end{algorithm}

\begin{algorithm}[!t]
	\caption{Generalized-EPIC $(G(\psi),\vec{x},\vec{b},c,\varepsilon)$ \cite{huang2020efficient}}\label{a3}
	\begin{algorithmic}[1]
		\renewcommand{\algorithmicrequire}{\textbf{Input:}}
		\renewcommand{\algorithmicensure}{\textbf{Output:}}
		\REQUIRE A graph $G(\psi)=(V(\psi),E(\psi))$, the current seeding vector $\vec{x}\in\mathbb{Z}^V_+$, a vector $\vec{b}\in\mathbb{Z}^V_+$, a cost $c:V\times\mathbb{Z}_+\rightarrow\mathbb{R}_+$, and an error parameter $\varepsilon$
		\ENSURE An approximately optimal node $u^\circ\in V(\psi)$
		\STATE Initialize: $\delta:=0.01\cdot\varepsilon/|V(\psi)|$
		\STATE Initialize: $\bar{\varepsilon}:=(\varepsilon-\delta\cdot|V(\psi)|)/(1-\delta\cdot|V(\psi)|)$
		\STATE Initialize: $\hat{\varepsilon}:=\bar{\varepsilon}/(1-\bar{\varepsilon})$
		\STATE Initialize: $i_{max}:=\left\lceil\log_2\frac{(2+2\cdot\hat{\varepsilon}/3)\cdot|V(\psi)|}{\hat{\varepsilon}^2}\right\rceil+1$ and $a=\ln\left(\frac{2\cdot i_{max}}{\delta}\right)$
		\STATE Initialize: $\theta:=\ln\left(\frac{2}{\delta}\right)+\ln\left(\tbinom{|V(\psi)|}{1}\right)$
		\STATE Generate two collections $\mathcal{R}_1(\psi)$ and $\mathcal{R}_2(\psi)$ of random RR-sets with $|\mathcal{R}_1(\psi)|=|\mathcal{R}_2(\psi)|=\theta$
		\FOR {$i=1$ to $i_{max}$}
		\STATE $u^\circ\in\arg\max_{u\in V(\psi),\vec{x}(u)<\vec{b}(u)}H_{\mathcal{R}_1(\psi)}(\{u\}|\vec{x})$
		\STATE $H^u(\{u^*\})\leftarrow H_{\mathcal{R}_1(\psi)}(\{u^\circ\}|\vec{x})$
		\STATE $H^l(\{u^\circ\})\leftarrow\left(\sqrt{H_{\mathcal{R}_2(\psi)}+\frac{2\cdot a}{9\cdot|\mathcal{R}_2(\psi)|}}-\sqrt{\frac{a}{2\cdot|\mathcal{R}_2(\psi)|}}\right)^2-\frac{a}{18\cdot|\mathcal{R}_2(\psi)|}$
		\IF {$\frac{H^l(\{u^\circ\})}{H^u(\{u^*\})}\geq1-\bar{\varepsilon}$ or $i=i_{max}$}
		\RETURN $u^\circ$
		\ENDIF
		\STATE Double the size of $\mathcal{R}_1(\psi)$ and $\mathcal{R}_2(\psi)$ with new random RR-sets
		\ENDFOR
	\end{algorithmic}
\end{algorithm}

\subsection{Theoretical Analysis and Time Complexity}
According to the current seeding vector $\vec{x}$ and its partial relization $\psi$ at the beginning of each iteration (line 5 of Algorithm \ref{a2}), we can get a subgraph $G(\psi)$ and a collection of random RR-sets $\mathcal{R}(\psi)$. Now, we define a function $H_{\mathcal{R}(\psi)}(\{u\}|\vec{x})=\beta_u\cdot F_{\mathcal{R}(\psi)}(\{u\})/c(\langle u,\vec{x}(u)+1\rangle)$, thereby the $|V(\psi)|\cdot H_{\mathcal{R}(\psi)}(\{u\}|\vec{x})$ is an unbiased estimator of $\Delta(u|\vec{x},\psi)/c(\langle u,\vec{x}(u)+1\rangle)$. Next, a natural question is how to determine the number of RR-sets in $\mathcal{R}(\psi)$. The proceduce of generating enough random RR-sets of $G(\psi)$ and returning the approximately optimal node $u^\circ\in V(\psi)$ (line 7 of Algorithm \ref{a2}) in each iteration is shown in Algorithm \ref{a3}. It is adapted from the sampling process of EPIC in \cite{huang2020efficient}, but there are several differences: (1) The seed size is fixed to one; and (2) The targeted estimator is the function $H_{\mathcal{R}(\psi)}(\{u\}|\vec{x})$ we defined before instead of $F_{\mathcal{R}(\psi)}(\{u\})$. Thus, the sampling process shown as algorithm \ref{a3} is called ``Generalized-EPIC''.

From line 1 to line 5 of Algorithm \ref{a3}, it initializes those parameters similar to EPIC in \cite{huang2020efficient} but fixs the seed set to one, then generate two collections $\mathcal{R}_1(\psi)$ and $\mathcal{R}_2(\psi)$ of random RR-sets with the same size. In each iteration, it select the feasible node $u^\circ\in V(\psi)$ that maximizes the estimator $H_{\mathcal{R}_1(\psi)}(\{u\}|\vec{x})$, which can be computed in polynomial time. Denoted by $u^*$ the optimal feasible node that maximizes the unit marginal gain $\Delta(u|\vec{x},\psi)/c(\langle u,\vec{x}(u)+1\rangle)$, the $H^u(\{u^*\})$ is an uppper bound on $H_{\mathcal{R}_1(\psi)}(\{u^*\}|\vec{x})$. Thus, we have $H^u(\{u^*\})=H_{\mathcal{R}_1(\psi)}(\{u^\circ\}|\vec{x})\geq H_{\mathcal{R}_1(\psi)}(\{u^*\}|\vec{x})$. Moveover, the $|V(\psi)|\cdot H^l(\{u^\circ\})$ gives an accurate lower bound on $\Delta(u^\circ|\vec{x},\psi)/c(\langle u^\circ,\vec{x}(u^\circ)+1\rangle)$ with high probability. After that, it checks whether the stopping condition in line 11 can be satisfied. If true, it will return an approximate optimal node $u^\circ$ definitely.

\begin{lem}\label{lem3}
	Given the current seeding vector $\vec{x}$ and its partial realization $\psi$, the feasible node $u^\circ$ returned by Algorithm \ref{a3} achieves a $(1-\varepsilon)$ expected approximation guarantee within $O((|V(\psi)|+|E(\psi)|)\cdot(\log(|V(\psi)|)+\log(1/\varepsilon))/\varepsilon^2)$ expected time. That is
	\begin{equation}
	\mathbb{E}_\omega\left[\frac{\Delta(u^\circ|\vec{x},\psi)}{c(\langle u^\circ,\vec{x}(u^\circ)+1\rangle)}\right]\geq(1-\varepsilon)\cdot\max_{u\in V(\psi),\vec{x}(u)<\vec{b}(u)}\left\{\frac{\Delta(u|\vec{x},\psi)}{c(\langle u,\vec{x}(u)+1\rangle)}\right\}
	\end{equation}
\end{lem}
\begin{proof}
	Given the current seeding vector $\vec{x}$ and its partial realization $\psi$, let us look at the targeted function $H_{\mathcal{R}(\psi)}(\{u\}|\vec{x})=\beta_u\cdot F_{\mathcal{R}(\psi)}(\{u\})/c(\langle u,\vec{x}(u)+1\rangle)$. It is a weighted coverage on the collection $\mathcal{R}(\psi)$, where we can consider the $\beta_u/c(\langle u,\vec{x}(u)+1\rangle)$ as the weight of each node $u\in V(\psi)$. The weighted coverage function is submodular, thereby we can compute the node $u^\circ\in\arg\max_{u\in V(\psi),\vec{x}(u)<\vec{b}(u)}H_{\mathcal{R}_1(\psi)}(\{u\}|\vec{x})$ accurately shown as line 8 of Algorithm \ref{a3} in polynomial time. Because of its submodularity, Lemma \ref{lem3} can be obtained by adapting from the expected approximation guarantee of EPIC in \cite{huang2020efficient}.
\end{proof}

Let us look back at Algorithm \ref{a2}. The actual number of activated seeds should be much less than the number of actual iterations in Algorithm \ref{a2}, since there are some iterations that fail to activate its selected node. Based on that, we can make the following assumptions:
	\begin{enumerate}
	\item Generate an active seed successfully in each iteration, namely we suppose $\beta_u=1$ for each node $u\in V$.
	\item The node we select in each itertaion has the lowest cost until now.
	\item We sort the node set $V$ as $\{v'_1,v'_2,\cdots,v'_n\}$ with $c(\langle v'_1,1\rangle)\leq c(\langle v'_2,1\rangle)\leq\cdots\leq c(\langle v'_n,1\rangle)$.
	\end{enumerate}
Given a graph $G=(V,E)$ and a budget $k$, we can define the maximum number of iterations in Algorithm \ref{a2} as $r$. That is
\begin{equation}\label{eq41}
r=
\begin{cases}
n &\text{if }\sum_{i=1}^{n}c(\langle v'_i,1\rangle)\leq k\\
q &\text{else }q=\min\{q|\sum_{i=1}^{q}c(\langle v'_i,1\rangle)\geq k\}
\end{cases}
\end{equation}
By finding the smallest $r$ such that $\sum_{i=1}^{r}c(\langle v'_i,1\rangle)\geq k$, it is obvious that the actual number of iterations in Algorithm \ref{a2} must be less than $r$ defined in Equation (\ref{eq41}).
\begin{thm}\label{thm4}
	The sampled adaptive greedy policy $\pi^{gs}(\kappa,\omega)$ shown as Algorithm \ref{a2} achieved a $(1-e^{-1+\varepsilon})$ expected approximation guarantee within $O(r\cdot(n+m)\cdot(\log(n)+\log(1/\varepsilon))/\varepsilon^2)$ expected time. Thus, for any policy $\pi^*(\kappa)$ that satisfies $\eta(\pi^*(\kappa),\phi)\leq\vec{b}$ and $\mathbb{E}_\kappa[c(\eta(\pi^*(\kappa),\phi))]\leq k$ for any realization $\phi$, we have
	\begin{equation}
		\mathbb{E}_\kappa\left[\mathbb{E}_\omega\left[f_{avg}(\pi^{gs}(\kappa,\omega))\right]\right]\geq\left(1-e^{-1+\varepsilon}\right)\cdot\mathbb{E}_\kappa\left[f_{avg}(\pi^*(\kappa))\right]
	\end{equation}
\end{thm}
\begin{proof}
	According to the above assumptions, the maximum number of iterations can be executed in Algorithm \ref{a2} is $r$. Based on Lemma \ref{lem3}, the selected node in each iteartion satisfies $(1-\varepsilon)$ expected approximation. In this extreme case, the total expected error over all iterations is $\varepsilon=(1/r)\cdot\sum_{i=1}^{r}\varepsilon$. Actually, the total expected error will be much less than $\varepsilon$ due to the $\beta_u\leq 1$ for each node $u\in V$. Here, there is no node can be activated in many iterations. Based on Theorem \ref{thm3} and Lemma \ref{lem3}, Theorem \ref{thm4} holds by inferring from Theorem 6 in \cite{huang2020efficient}.
\end{proof}
\begin{table}[h]
	\renewcommand{\arraystretch}{1.3}
	\caption{The statistics of four datasets in our simulations $(K=10^3)$}
	\label{table1}
	\centering
	\begin{tabular}{|c|c|c|c|c|}
		\hline
		\bfseries Dataset & \bfseries n & \bfseries m & \bfseries Type & \bfseries Avg. Degree\\
		\hline
		NetScience & 0.4K & 1.01K & undirected & 5.00\\
		\hline
		Wiki & 1.0K & 3.15K & directed & 6.20\\
		\hline
		HetHEPT & 12.0K & 118.5K & undirected & 19.8\\
		\hline
		Epinions & 75.9K & 508.8K & directed & 13.4\\
		\hline
	\end{tabular}
\end{table}

\section{Experiment}
In this section, we carry out several experiments on different datasets to validate the performance of our proposed policy. It aims to test the efficiency of our sampled adaptive greedy policy shown as Algorithm \ref{a2} and its effectiveness compared to other adaptive heuristic policies. All of our experiments are programmed by python and run on a Windows machine with a 3.40GHz, 4 core Intel CPU and 16GB RAM.

\subsection{Dataset Description and Statistics}
There are four datasets used in our experiments: (1) NetScience \cite{nr}: a co-authorship network, co-authorship among scientists to publish papers about network science; (2) Wiki \cite{nr}: a who-votes-on-whom network, which come from the collection Wikipedia voting; (3) HetHEPT \cite{snapnets}: an academic collaboration relationship on high energy physics area; and (4) Epinions \cite{snapnets}: a who-trust-whom online social network on Epinions.com, a general consumer review site. The statistics information of these four datasets is represented in Table \ref{table1}. For the undirected graph, each undirected edge is replaced with two reversed directed edges.

\begin{algorithm}[!t]
	\caption{\text{Greedy $(G,\mu,k,\vec{b},c)$}}\label{a4}
	\begin{algorithmic}[1]
		\renewcommand{\algorithmicrequire}{\textbf{Input:}}
		\renewcommand{\algorithmicensure}{\textbf{Output:}}
		\REQUIRE A graph $G=(V,E)$, a function $\mu(\vec{x})$, a budget $k\in\mathbb{R}_+$, a vector $\vec{b}\in\mathbb{Z}^V_+$ and, a cost function $c:V\times\mathbb{Z}_+\rightarrow\mathbb{R}_+$
		\ENSURE A seeding vector $\vec{x}\in\mathbb{Z}^V_+$ and $\mu(\vec{x})$
		\STATE Initialize: $\vec{x}:=\vec{0}$
		\WHILE {$c(\vec{x})<k$}
		\STATE $u^*\in\arg\max_{u\in V,\vec{x}(u)<\vec{b}(u)}(\mu_G(\vec{x}+\vec{e}_u)-\mu_G(\vec{x}))/c(\langle u,\vec{x}(u)+1\rangle)$
		\IF {$c(\vec{x})+c(\langle u^*,\vec{x}(u^*)+1\rangle)>k$}
		\STATE $\textbf{break}$ with probability $1-(k-c(\vec{x}))/c(\langle u^*,\vec{x}(u^*)+1\rangle)$
		\ENDIF
		\STATE $\vec{x}(u^*):=\vec{x}(u^*)+1$
		\ENDWHILE
		\RETURN $\vec{x}$, $\mu(\vec{x})$
	\end{algorithmic}
\end{algorithm}

\subsection{Experimental Setting}
The diffusion model used in our experiments relies on the IC-model. For each edge $(u,v)\in E$, we set $p_{uv}=1/|N^-(v)|$, which is widely used by prior works about influence maximization \cite{kempe2003maximizing} \cite{borgs2014maximizing} \cite{tang2014influence} \cite{tang2015influence} \cite{nguyen2016stop}. There are several parameters associated with the objective function of our Adaptive-IMMA problem. Here, we set the vector $\vec{b}=\{5\}^V$ where each node can be attempted to activate as a seed at most $5$ times; the cost of each trial $c(\langle u,1\rangle)=1$ and $c(\langle u,i+1\rangle)=1.2\times c(\langle u,i\rangle)$; and variable budget $k\in\{0,10,20,30,40,50\}$. Besides, for each node $u\in V$, its probability $\beta_u$ is sampled from a normal distribution within given a mean, variance and interval. For each adaptive policies, we generate 20 realizations (test it 20 times) randomly and take the average of their results as its final performance.

We perform two experiments with different purposes in this section. The first experiment is to test the time efficiency of the adaptive greedy policy and sampled adaptive greedy policy (Algorithm \ref{a2}), then validate the superiority over their non-adaptive settings. The corresponding non-adaptive versions of adaptive greedy policy and sampled adaptive greedy policy are referred to as greedy and sampled greedy algorithm respectively. Here, the greedy algorithm and adaptive greedy policy are implemented by MC simulations. They can be shown as follows: 
\begin{enumerate}
	\item Greedy algorithm: Shown as Algorithm \ref{a4}, it selects a node $u\in V$ with $\vec{x}(u)<\vec{b}$ such that maximizes the unit marginal gain $(\mu_G(\vec{x}+\vec{e}_u)-\mu_G(\vec{x}))/c(\langle u,\vec{x}(u)+1\rangle)$ in each iteration. The selected node in the last iteration will be contained with a probability. To estimate the value of $\mu_G(\vec{x})$, we have
	\begin{equation}
		\mu_G(\vec{x})=\sigma_{\widetilde{G}}(\widetilde{V}-V)-|V|
	\end{equation}
	where we need to create a constructed graph $\widetilde{G}=(\widetilde{V},\widetilde{E})$ by adding a new node $\widetilde{u}$ and a new directed edge $(\widetilde{u},u)$ for each node $u\in V$ to $G$, where $(\widetilde{u},u)$ is with activation probability $p_{\widetilde{u}u}=1-(1-\beta_u)^{\vec{x}(u)}$. Here, the $\sigma_{\widetilde{G}}(\widetilde{V}-V)$ can be estimated by MC simulations, which is an effective methods to estimate the value of $\mu(\vec{x})$ \cite{guo2020continuous}.
	\item Adaptive greedy policy: Shown as Algorithm \ref{a1}, we can compute the unit marginal gain $\Delta(u|\vec{x},\psi)/c(\langle u,\vec{x}(u)+1\rangle)$ through $\beta_u\cdot\sigma_{G(\psi)}(\{u\})$ according to Equation (\ref{eq38}), where $\sigma_{G(\psi)}(\{u\})$ can be estimated by MC simulations.
	\item Sampled greedy algorithm: Here, we require to obtain an unbiased estimator of $\mu(\vec{x})$. Let $\mathcal{R}$ be a collection of random RR-sets sampled from $G$, we have
	\begin{equation}\label{eq44}
		\mu_G(\vec{x})=|V|\cdot\mathbb{E}_R\left[1-\prod_{u\in R}(1-\beta_u)^{\vec{x}(u)}\right]
	\end{equation}
	Let $F_{\mathcal{R}}(\vec{x})=(\theta-\sum_{i=1}^{\theta}\prod_{u\in R_i}(1-\beta_u)^{\vec{x}(u)})/\theta$, thereby we have $|V|\cdot F_{\mathcal{R}}(\vec{x})$ is an unbiased estimator of $\mu(\vec{x})$. Because there is no existing algorithm to determine the number of random RR-sets in this case, we will guess a size of $\mathcal{R}$ according to datasets and budgets. Given a collection $\mathcal{R}$, it selects a node $u^\circ\in V$ with $\vec{x}(u)<\vec{b}(u)$ such that maximizes the unit marginal coverage $(F_{\mathcal{R}}(\vec{x}+\vec{e}_u)-F_{\mathcal{R}}(\vec{x}))/c(\langle u,\vec{x}(u)+1\rangle)$ in each iteration. The selected node in the last iteration will be contained with a probability, which is similar to Algorithm \ref{a4}.
	\item Sampled adaptive greedy policy: It can be implemented by Algorithm \ref{a2} with the error parameter $\varepsilon=0.5$.
\end{enumerate}

The second experiment is to test the performance of our sampled adaptive greedy policy compared with other heuristic adaptive policies, which aims to evaluate its effectiveness. The difference between these heuristic adaptive policies and our sampled adaptive greedy policy lies in how to select a node $u^\circ$ from the feasible node set that satisfies $u\in V(\psi)$ and $\vec{x}(u)<\vec{b}(u)$ in each iteration. Thus, the only difference is in line 6 of Algorithm \ref{a2} and other procedures are totally identical. In other words, they are obtained by replacing line 6 of Algorithm \ref{a2} with these heuristic strategies, summarized as follows: (1) Random: select a node $u^\circ$ from the feasible node set uniformly in each iteration; (2) MaxDegree: select a node $u^\circ$ from the feasible node set that maximizes $N^+(u)/c(\langle u,\vec{x}(u)+1\rangle)$  in each iteration; (3) MaxProb: select a node $u^\circ$ from the feasible node set that maximizes $\beta_u/c(\langle u,\vec{x}(u)+1\rangle)$ in each iteration; and (4) MaxDegreeProb: select a node $u^\circ$ from the feasible node set that maximizes $\beta_u\cdot N^+(u)/c(\langle u,\vec{x}(u)+1\rangle)$ in each iteration.

\begin{figure}[!t]
	\centering
	\subfigure[NetScience, $\beta\sim N(0.4,1)$]{
		\includegraphics[width=0.49\columnwidth]{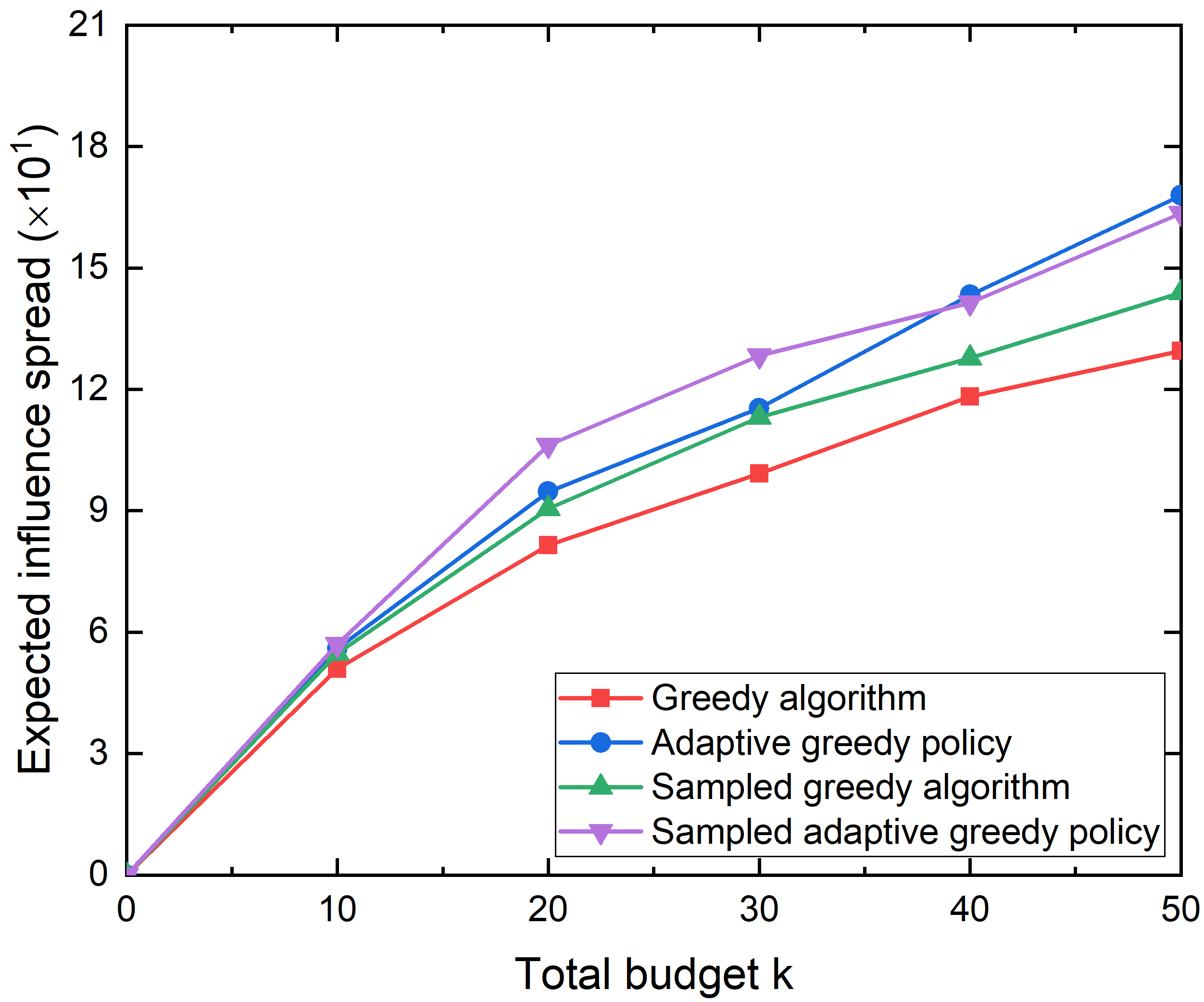}
		%\caption{fig1}
	}%
	\subfigure[NetScience, $\beta\sim N(0.6,1)$]{
		\centering
		\includegraphics[width=0.49\columnwidth]{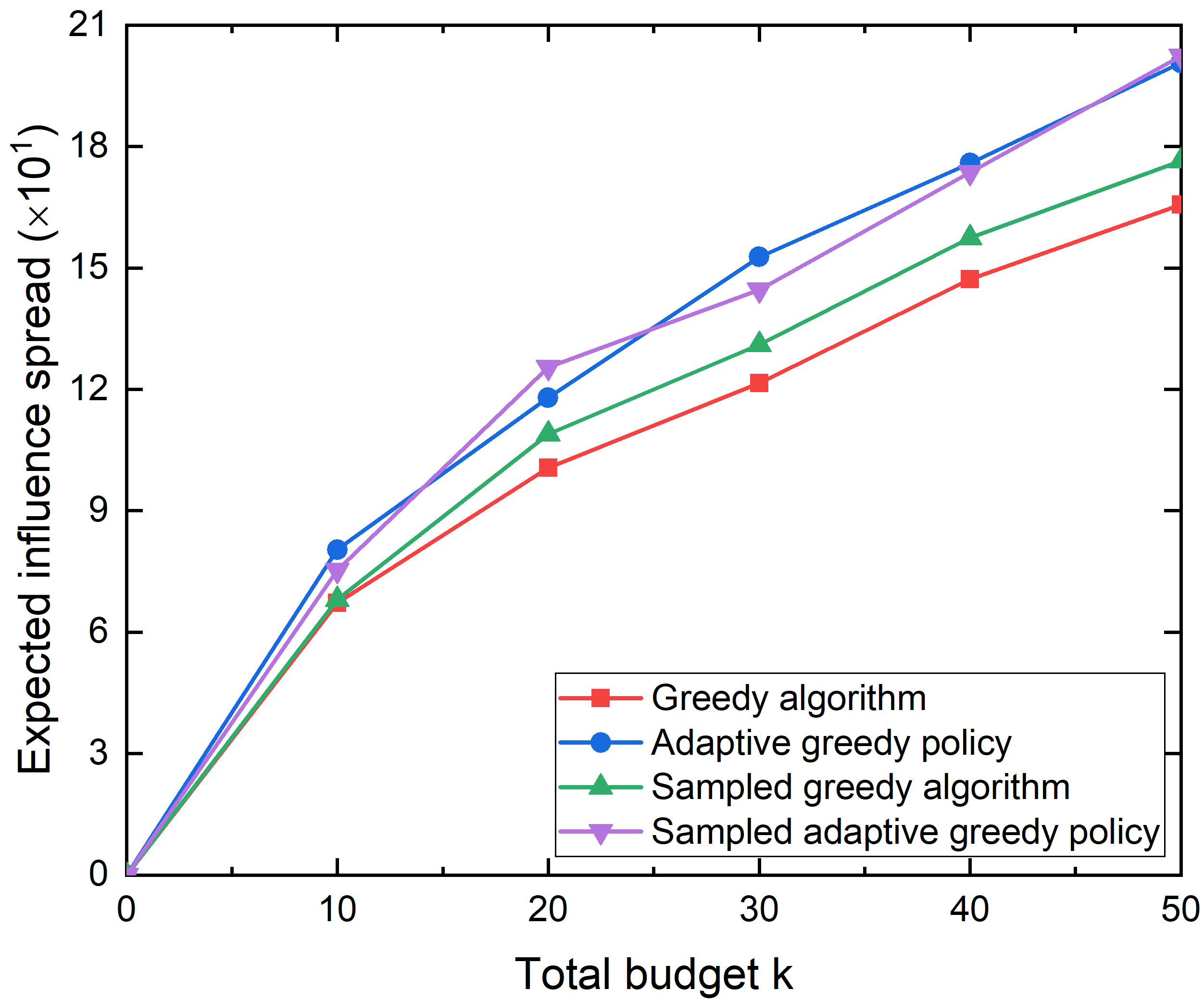}
		%\caption{fig2}
	}%
	
	\subfigure[Wiki, $\beta\sim N(0.4,1)$]{
		\centering
		\includegraphics[width=0.49\columnwidth]{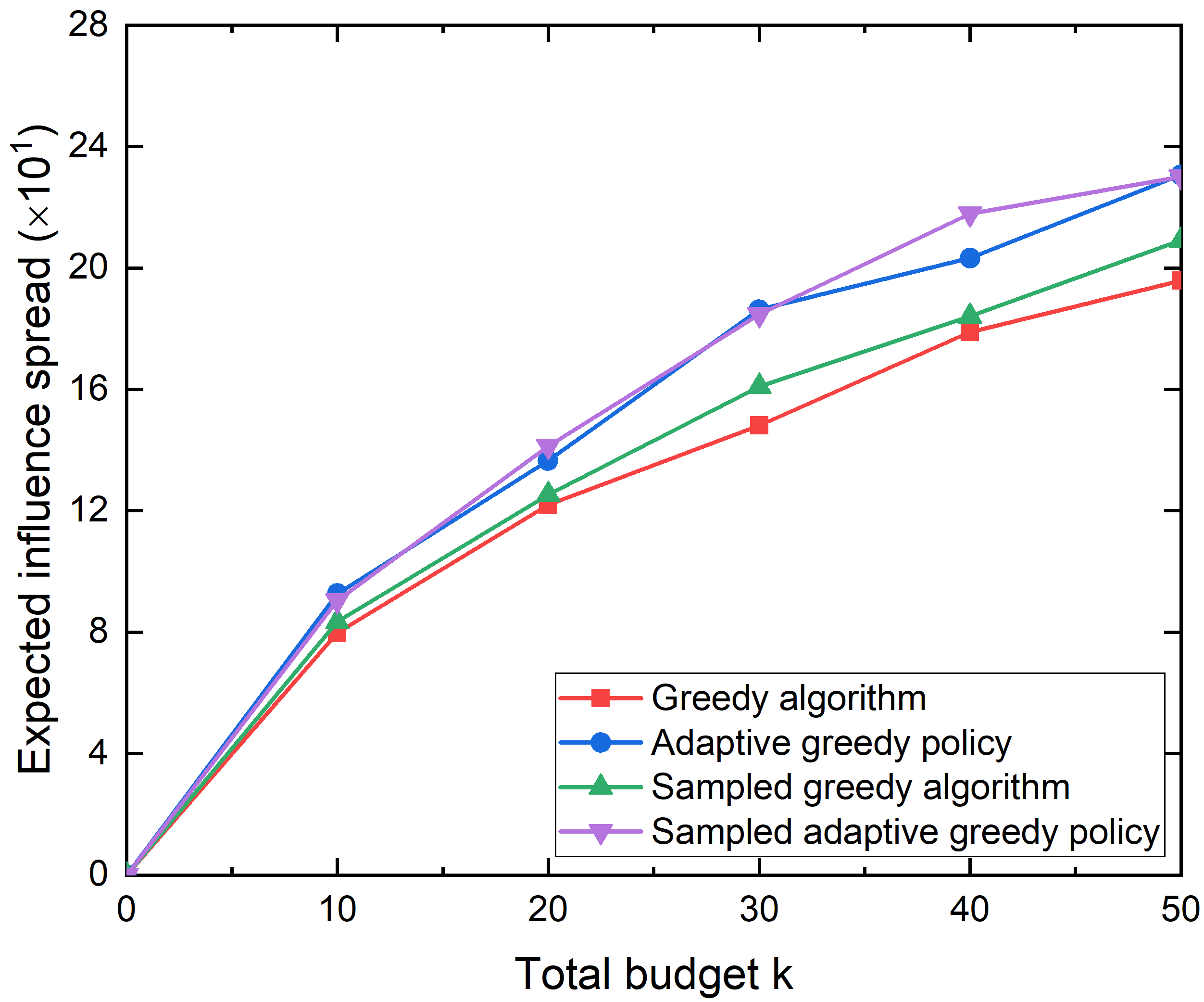}
		%\caption{fig2}
	}%
	\subfigure[Wiki, $\beta\sim N(0.6,1)$]{
		\centering
		\includegraphics[width=0.49\columnwidth]{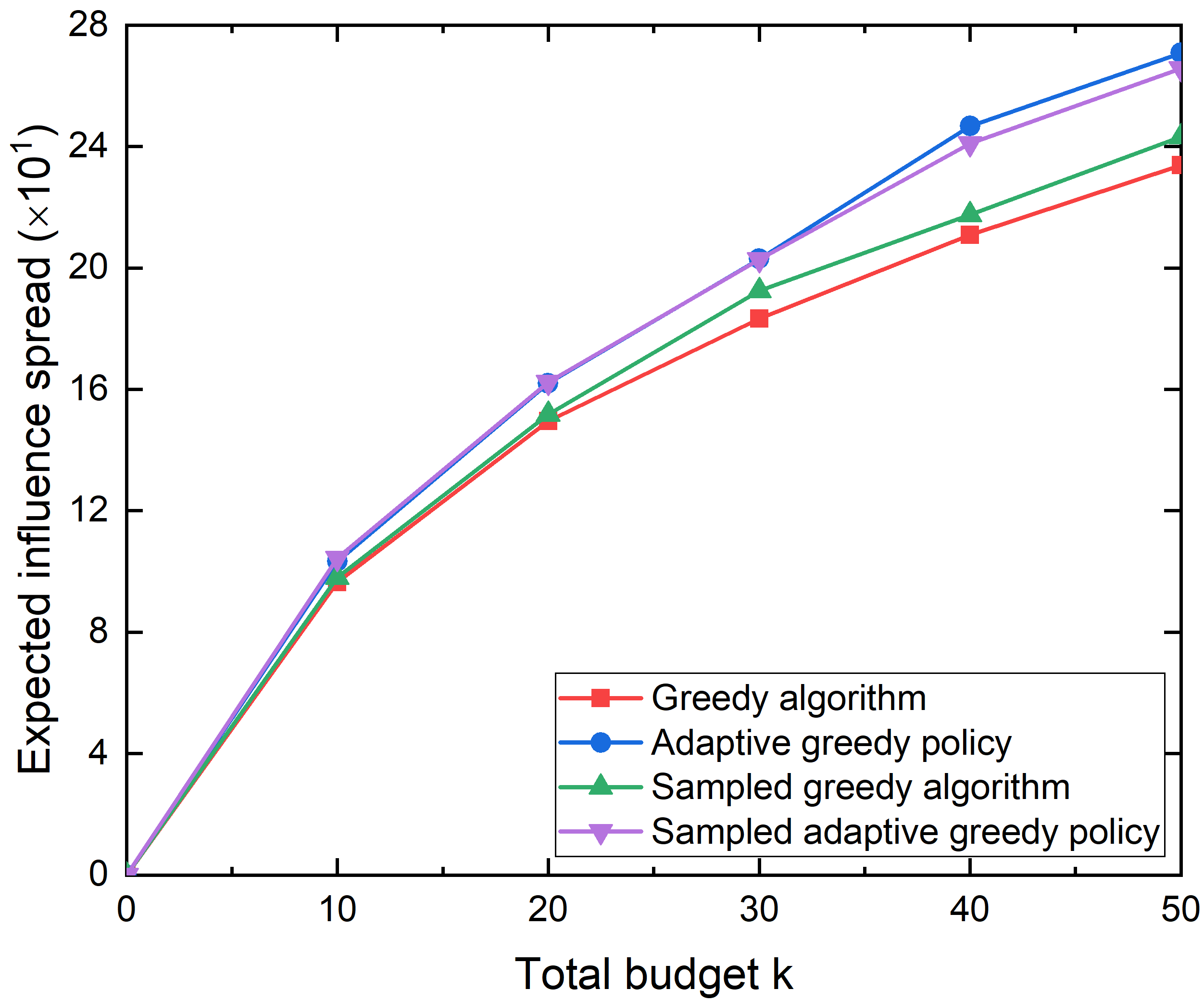}
		%\caption{fig2}
	}%
	\centering
	\caption{The expected influence spread achieved by the (sampled) greedy algorithm and (sampled) adaptive greedy policy under the NetScience and Wiki datasets.}
	\label{fig1}
\end{figure}

\begin{figure}[!t]
	\centering
	\subfigure[NetScience, $\beta\sim N(0.4,1)$]{
		\includegraphics[width=0.49\columnwidth]{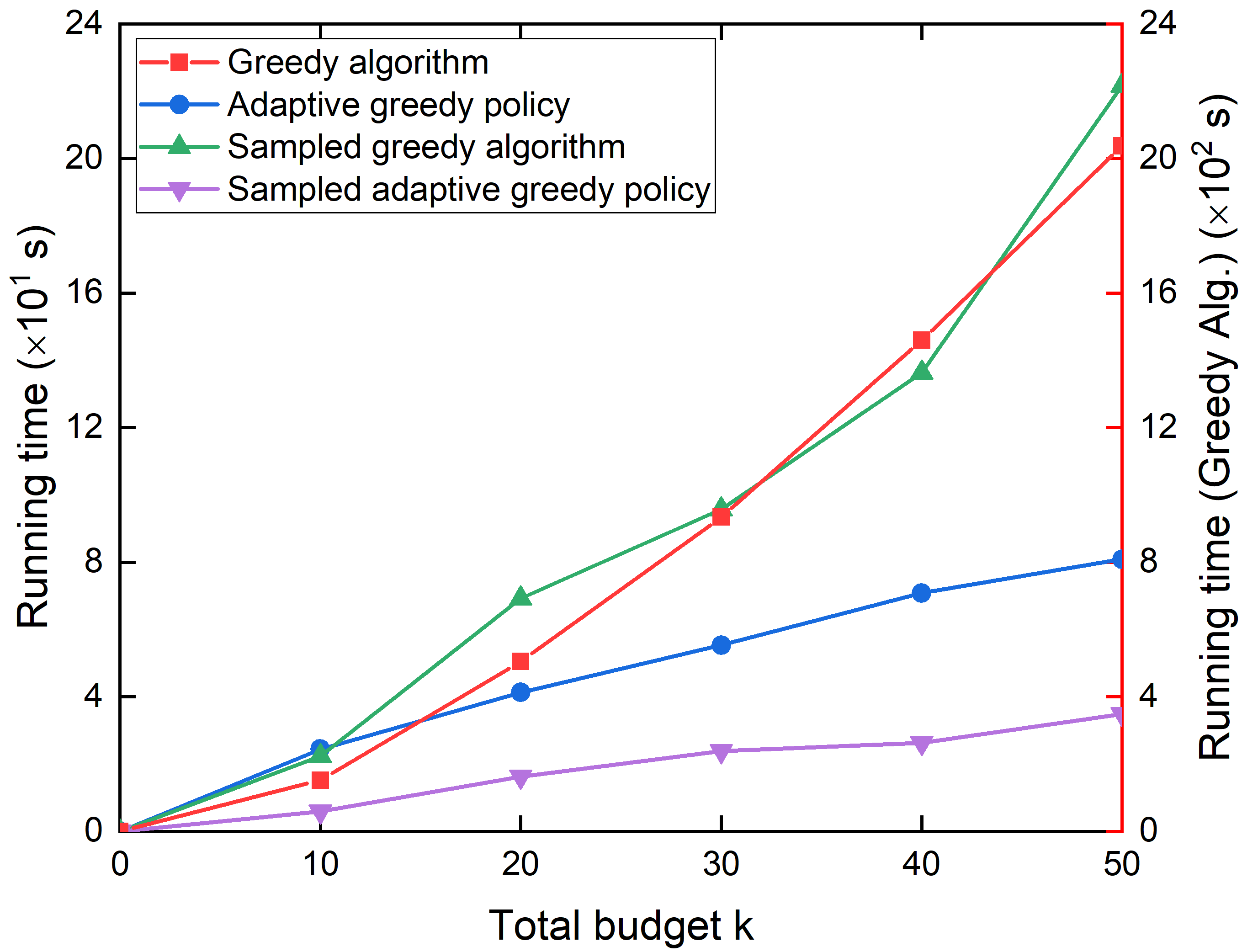}
		%\caption{fig1}
	}%
	\subfigure[NetScience, $\beta\sim N(0.6,1)$]{
		\centering
		\includegraphics[width=0.49\columnwidth]{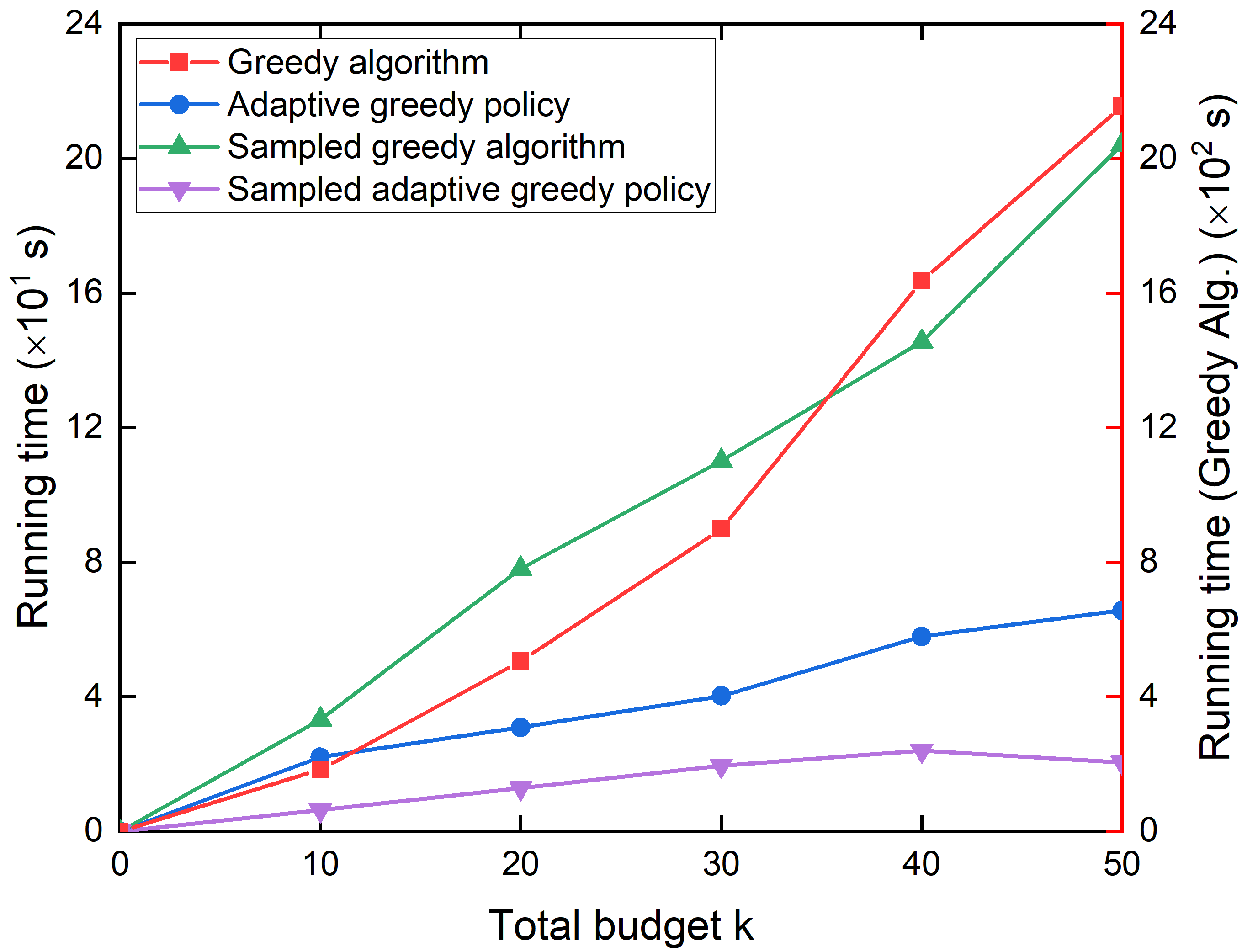}
		%\caption{fig2}
	}%
	
	\subfigure[Wiki, $\beta\sim N(0.4,1)$]{
		\centering
		\includegraphics[width=0.49\columnwidth]{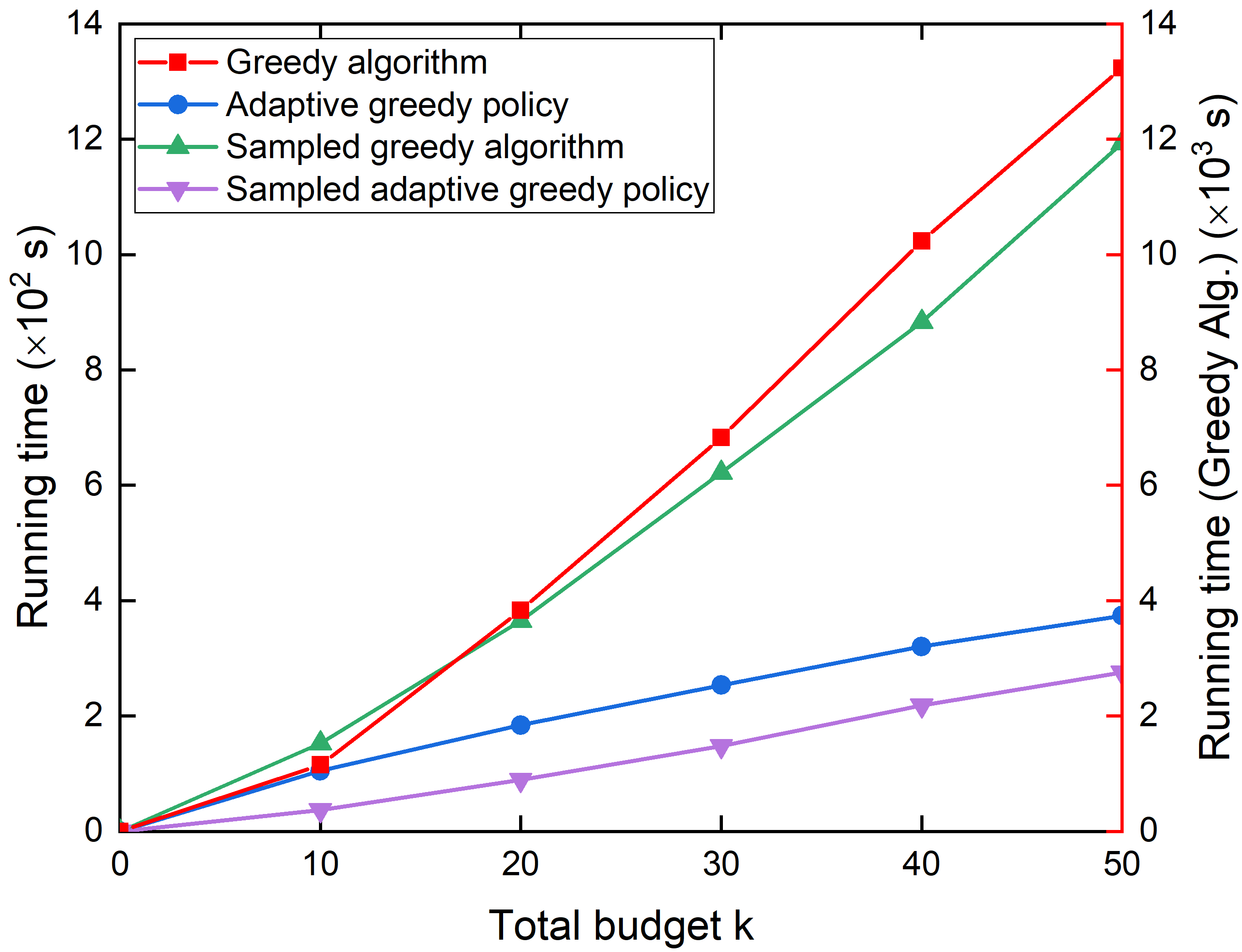}
		%\caption{fig2}
	}%
	\subfigure[Wiki, $\beta\sim N(0.6,1)$]{
		\centering
		\includegraphics[width=0.49\columnwidth]{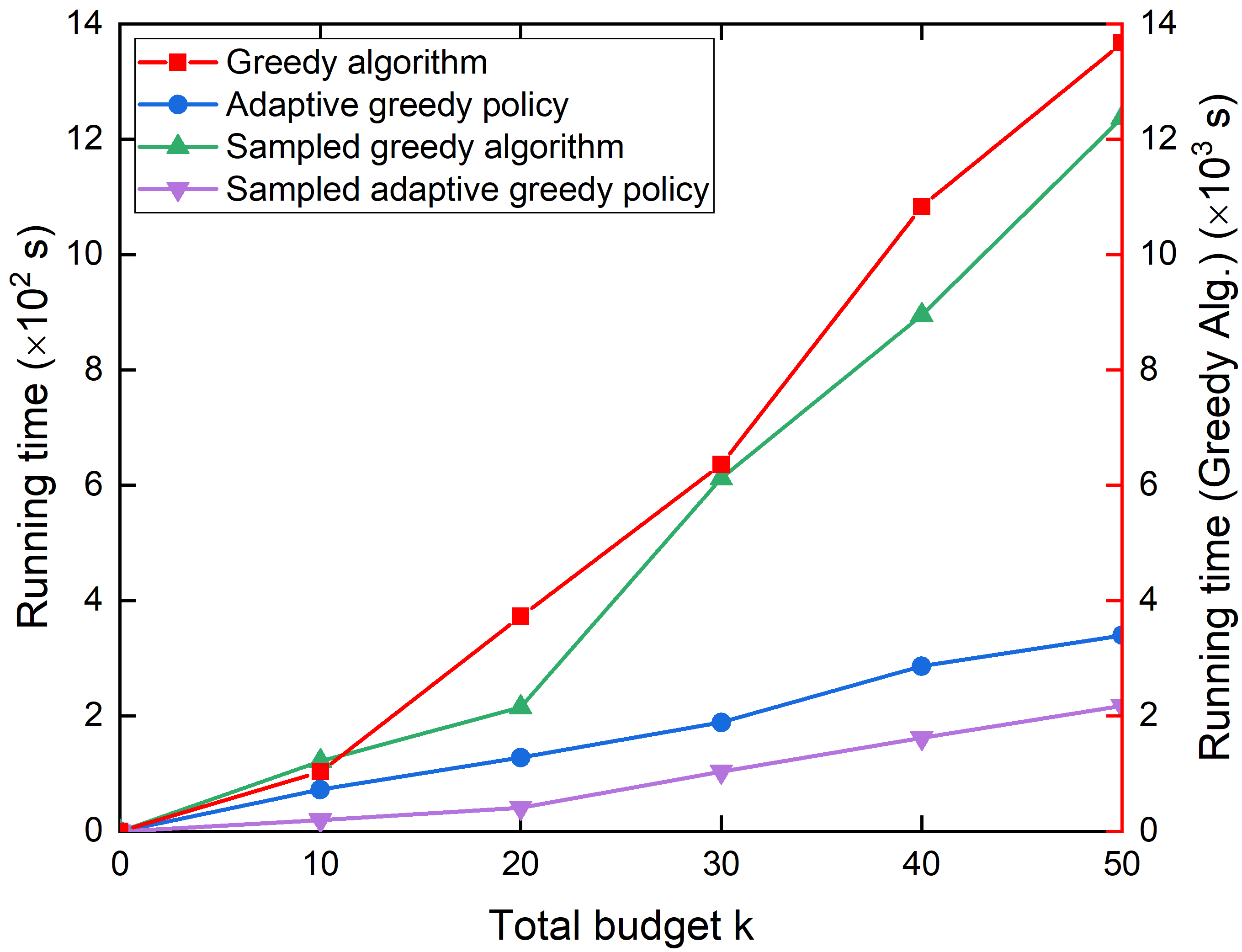}
		%\caption{fig2}
	}%
	\centering
	\caption{The running time achieved by the (sampled) greedy algorithm and (sampled) adaptive greedy policy under the NetScience and Wiki datasets.}
	\label{fig2}
\end{figure}

\begin{figure}[!t]
	\centering
	\subfigure[NetScience, $\beta\sim N(0.5,1)$]{
		\includegraphics[width=0.49\columnwidth]{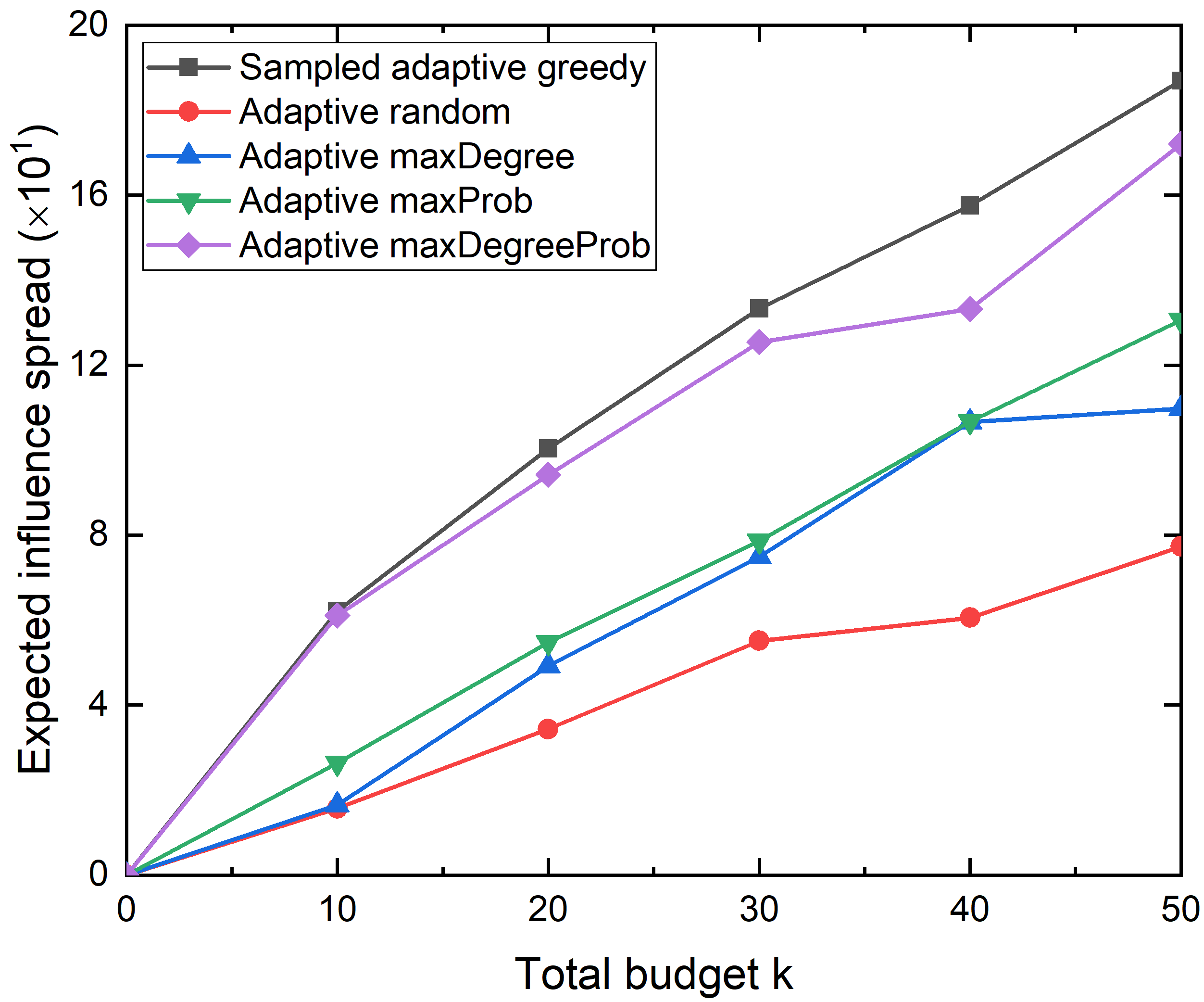}
		%\caption{fig1}
	}%
	\subfigure[Wiki, $\beta\sim N(0.5,1)$]{
		\centering
		\includegraphics[width=0.49\columnwidth]{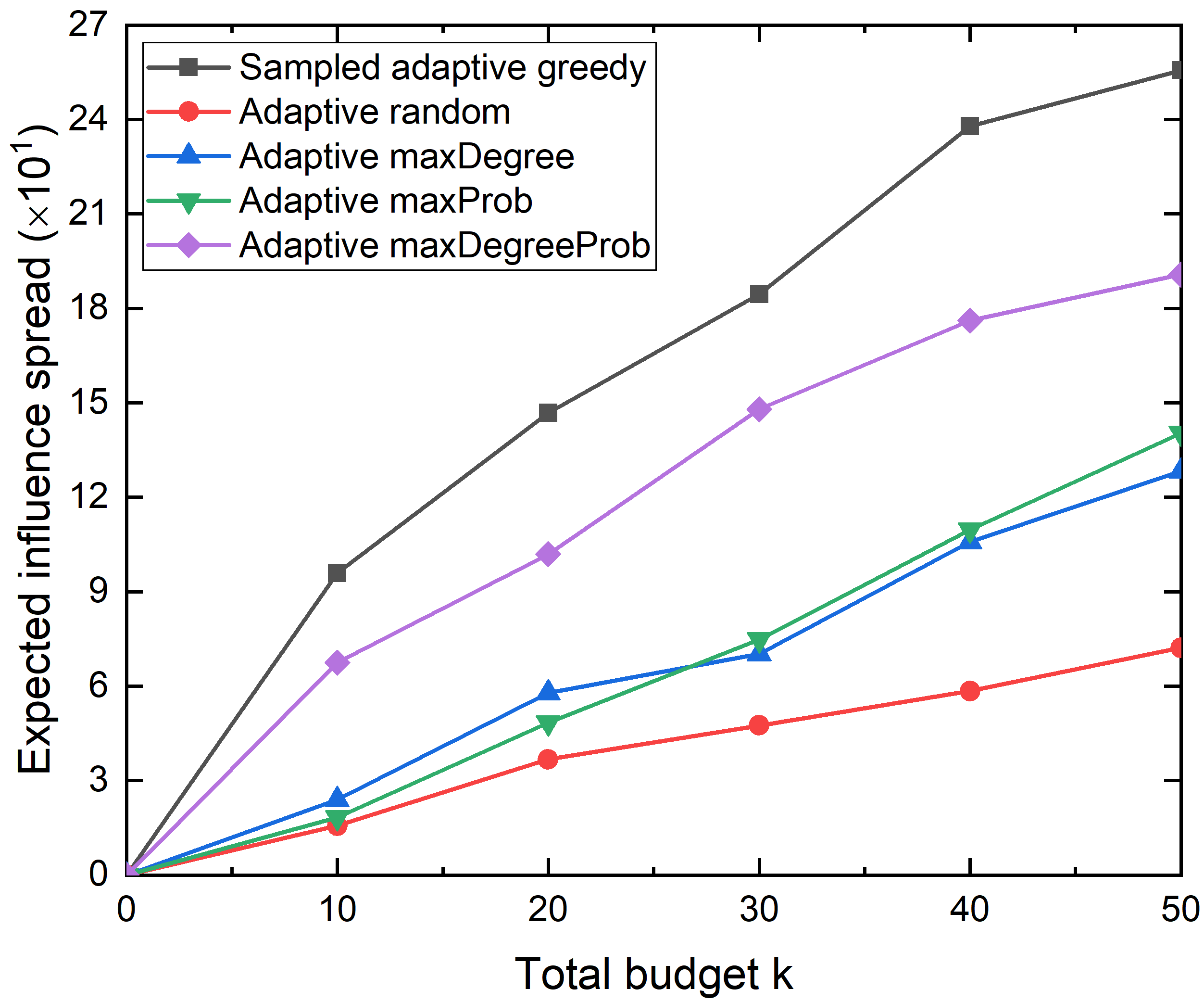}
		%\caption{fig2}
	}%
	
	\subfigure[HetHEPT, $\beta\sim N(0.5,1)$]{
		\centering
		\includegraphics[width=0.49\columnwidth]{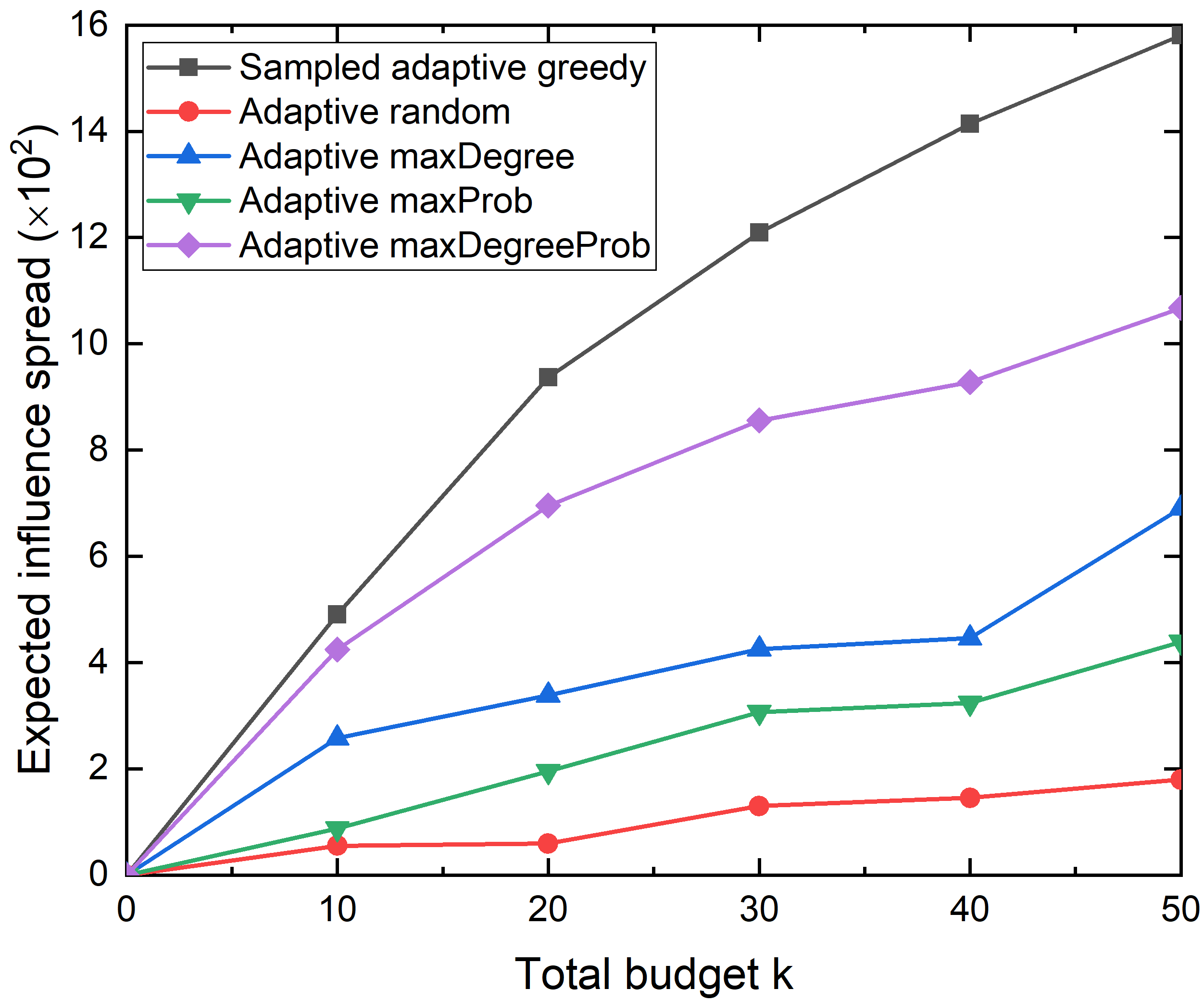}
		%\caption{fig2}
	}%
	\subfigure[Epinions, $\beta\sim N(0.5,1)$]{
		\centering
		\includegraphics[width=0.49\columnwidth]{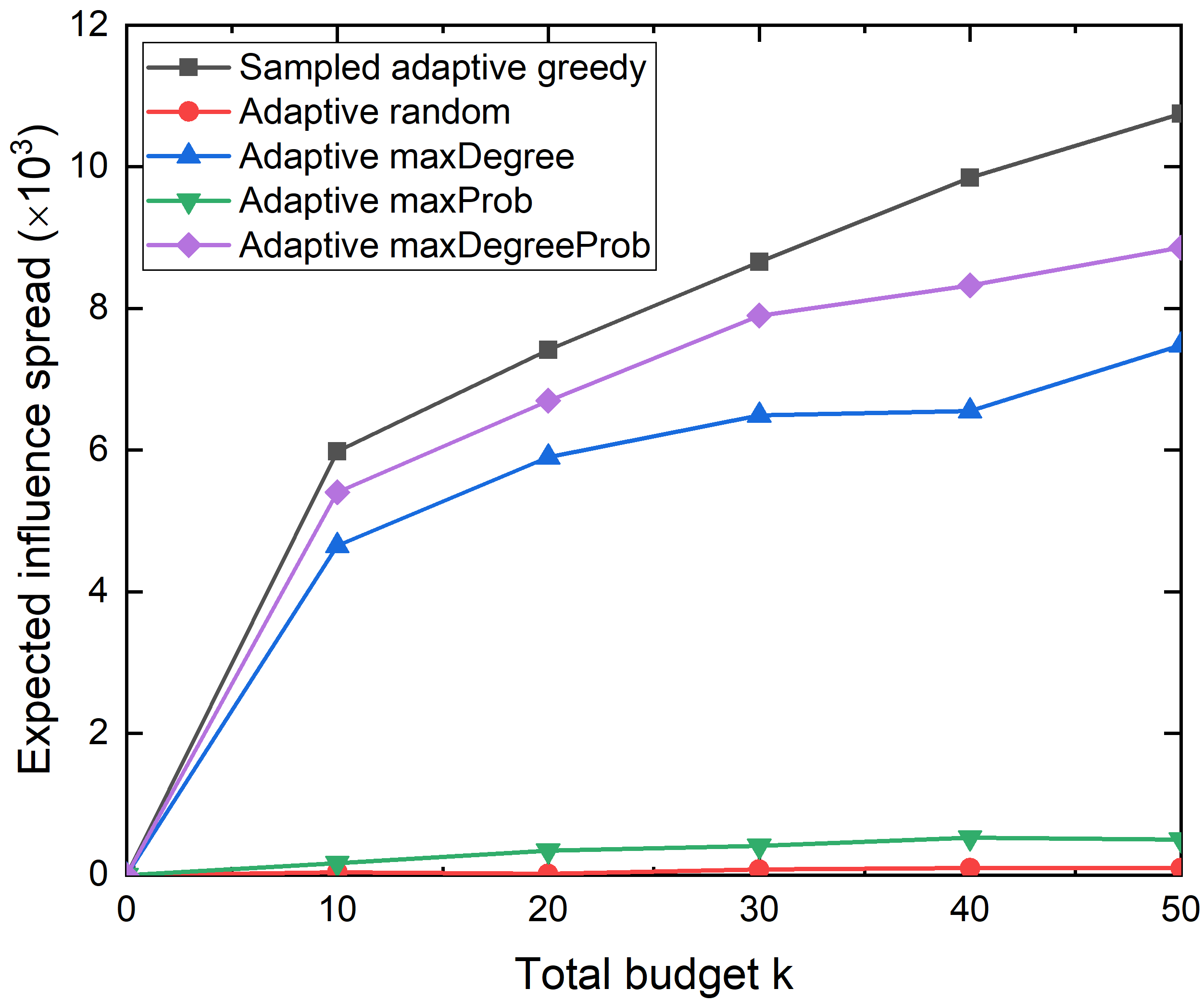}
		%\caption{fig2}
	}%
	\centering
	\caption{The performance comparisons between our sampled adaptive greedy policy and other heuristic adaptive adaptive under the four datasets}
	\label{fig3}
\end{figure}

\subsection{Experimental Results}
Figure \ref{fig1} and Figure \ref{fig2} are the experimental results of the first experiment. Figure \ref{fig1} draws the expected influence spread achieved by the (sampled) greedy algorithm and (sampled) adaptive greedy policy under the NetScience and Wiki datasets. Here, the probability $\beta\sim N(a,b)$ means $\beta_u$ for each node $u\in V$ is sampled from a truncated normal distribution whose mean is $a$ and variance is $b$ within the interval $[0,1]$. Because the greedy algorithm and adaptive greedy policy are implemented by MC simulations, and its time complexity is too high, thereby we only use these two small graphs to test them in this experiment. Here, the number of MC simulations for each estimation is set to $300$ in NetScience dataset and $600$ in Wiki dataset. This is far from enough, just for performance comparison. For the sampled greedy algorithm, the number of random RR-sets is determined based on experience, where we give $|\mathcal{R}|=5000+1000\cdot(k/10)$ in NetScience dataset and $|\mathcal{R}|=10000+2000\cdot(k/10)$ in Wiki dataset.

We note that the expected influence spread obtained by the adaptive greedy policy and sampled adaptive greedy policy is very close, which proves the effectiveness of our sampling techniques. Under the non-adaptive settings, the performance achieved by the sampled greedy algorithm is better than that achieved by the greedy algorithm. This may be because the number of MC simulations we set for each estimation is not enough to get a precise estimation. Thus, we are more inclined to think the results obtained by the sampled greedy algorithm are more precise. Compare the performances shown as Figure \ref{fig1}, we find that the sampled adaptive greedy policy has an obvious advantage, which is much better than the sampled greedy algorithm. This illustrates the effectiveness of our proposed adaptive policy from one aspect. Besides, with the increase of the mean of $\beta$, there is no doubt the expected influence spread will increase. However, we observe an interesting phenomenon where the gap between the performance under the adaptive settings and non-adaptive settings seems to be shrinking. This is because the uncertainty of nodes, whether to be an active seed or not, decreases as the mean of $\beta$ increases, thereby reducing the advantage of our adaptive policies.

Figure \ref{fig2} draws the running time achieved by the (sampled) greedy algorithm and (sampled) adaptive greedy policy under the NetScience and Wiki datasets. Here, in order to compare the running time of different strategies, we do not use parallel acceleration in our implementations. We note that the running time of the sampled adaptive greedy policy is smaller than that of the sampled greedy algorithm, which is counter-intuitive. This looks unreasonable because the sampled greedy algorithm only needs to generate a collection of RR-sets once and selects seed nodes in one batch, but the sampled adaptive greedy policy has to generate a new collection of RR-sets in each iteration. Why does it happen? First, the estimator of $\mu(\vec{x})$ shown as Equation (\ref{eq44}) is more complicated than the estimator of $\Delta(u|\vec{x},\psi)$ shown as Equation (\ref{eq38}). Second, the number of RR-sets we give in the sampled greedy algorithm may be too much, which exceeds actual needs. At last, the sampling process will be faster and faster as the graph gets smaller in the sampled adaptive greedy policy. Then, we can see that the running time of the sampled adaptive greedy policy is less than that of the adaptive greedy policy even though the number of MC simulations is far from enough, which proves the efficiency of our sampling techniques. Compare to the running times achieved by the adaptive greedy policy, the greedy algorithm is very inefficient, nearly 10 times slower than the adaptive greedy policy. There are two reasons to explain this phenomenon. First, the graph that the adaptive greedy policy relies on is shrinking gradually as the number of iterations increases. Secondly, the process of reverse breadth-first search in MC simulations will be more time-consuming when the seed set is large.

Figure \ref{fig3} draws the performance comparisons between our sampled adaptive greedy policy and other heuristic adaptive policies under the four datasets. We can see that the expected influence spread of any adaptive policy increases with budget $k$ because attempting to select more seed results in a larger influence spread. The expected influence spread returned by our sampled adaptive greedy policy outperforms all other heuristic adaptive policies under any dataset, thereby its performance is the best undoubtedly. This illustrates the effectiveness of our proposed policy from another aspect. Among these heuristic adaptive policies, the adaptive maxDegreeProb policy has the largest expected influence spread, because it considers the node's degree and probability to be a seed comprehensively. The performance of other policies is unstable on different datasets. We can observe that the sampled adaptive greedy policy can obtain at least $10\%$ gain of the expected influence spread than the best heuristic adaptive policy. However, the gap between the sampled adaptive greedy policy and other heuristic adaptive policies can be affected by the dataset itself, since there are different topologies and graph realizations associated with different networks.

\section{conclusion}
In this paper, we have studied a variant of adaptive influence maximization, where the seed node we select may be unwilling to be the influencer and we can activate her many times. Because its objective function is defined on integer lattice, we propose the concepts of adaptive monotonicity on integer lattice and adaptive dr-submodularity firstly. Then, we summarize the properties of this problem and give a strict theoretical analysis about the approximation ratio of the adaptive greedy policy. Our approach can be used as a flexible framework to address adaptive monotone and dr-submodular function under the expected knapsack constraint. Combine with the-state-of-art EPIC algorithms, the sampled adaptive greedy policy is formulated, which reduces its running time significantly without losing the approximation guarantee. Eventually, we evaluate our proposed policies on four real networks and validate the effectiveness and efficiency comparing to their corresponding non-adaptive algorithms and other heuristic adaptive policies.

\begin{acks}
This work is partly supported by National Science Foundation under grant 1747818 and 1907472.
\end{acks}

\bibliographystyle{ACM-Reference-Format}
\bibliography{references}


\begin{thebibliography}{32}

%%% the Web address.

\ifx \showCODEN    \undefined \def \showCODEN     #1{\unskip}     \fi
\ifx \showDOI      \undefined \def \showDOI       #1{#1}\fi
\ifx \showISBNx    \undefined \def \showISBNx     #1{\unskip}     \fi
\ifx \showISBNxiii \undefined \def \showISBNxiii  #1{\unskip}     \fi
\ifx \showISSN     \undefined \def \showISSN      #1{\unskip}     \fi
\ifx \showLCCN     \undefined \def \showLCCN      #1{\unskip}     \fi
\ifx \shownote     \undefined \def \shownote      #1{#1}          \fi
\ifx \showarticletitle \undefined \def \showarticletitle #1{#1}   \fi
\ifx \showURL      \undefined \def \showURL       {\relax}        \fi
% invisible to TeX
\providecommand\bibfield[2]{#2}
\providecommand\bibinfo[2]{#2}
\providecommand\natexlab[1]{#1}
\providecommand\showeprint[2][]{arXiv:#2}

\bibitem[\protect\citeauthoryear{Borgs, Brautbar, Chayes, and Lucier}{Borgs
  et~al\mbox{.}}{2014}]%
        {borgs2014maximizing}
\bibfield{author}{\bibinfo{person}{Christian Borgs}, \bibinfo{person}{Michael
  Brautbar}, \bibinfo{person}{Jennifer Chayes}, {and} \bibinfo{person}{Brendan
  Lucier}.} \bibinfo{year}{2014}\natexlab{}.
\newblock \showarticletitle{Maximizing social influence in nearly optimal
  time}. In \bibinfo{booktitle}{\emph{Proceedings of the twenty-fifth annual
  ACM-SIAM symposium on Discrete algorithms}}. SIAM, \bibinfo{pages}{946--957}.
\newblock


\bibitem[\protect\citeauthoryear{Chen, Collins, Cummings, Ke, Liu, Rincon, Sun,
  Wang, Wei, and Yuan}{Chen et~al\mbox{.}}{2011}]%
        {chen2011influence}
\bibfield{author}{\bibinfo{person}{Wei Chen}, \bibinfo{person}{Alex Collins},
  \bibinfo{person}{Rachel Cummings}, \bibinfo{person}{Te Ke},
  \bibinfo{person}{Zhenming Liu}, \bibinfo{person}{David Rincon},
  \bibinfo{person}{Xiaorui Sun}, \bibinfo{person}{Yajun Wang},
  \bibinfo{person}{Wei Wei}, {and} \bibinfo{person}{Yifei Yuan}.}
  \bibinfo{year}{2011}\natexlab{}.
\newblock \showarticletitle{Influence maximization in social networks when
  negative opinions may emerge and propagate}. In
  \bibinfo{booktitle}{\emph{Proceedings of the 2011 siam international
  conference on data mining}}. SIAM, \bibinfo{pages}{379--390}.
\newblock


\bibitem[\protect\citeauthoryear{Chen, Lin, Tan, Zhao, and Zhou}{Chen
  et~al\mbox{.}}{2016}]%
        {chen2016robust}
\bibfield{author}{\bibinfo{person}{Wei Chen}, \bibinfo{person}{Tian Lin},
  \bibinfo{person}{Zihan Tan}, \bibinfo{person}{Mingfei Zhao}, {and}
  \bibinfo{person}{Xuren Zhou}.} \bibinfo{year}{2016}\natexlab{}.
\newblock \showarticletitle{Robust influence maximization}. In
  \bibinfo{booktitle}{\emph{Proceedings of the 22nd ACM SIGKDD International
  Conference on Knowledge Discovery and Data Mining}}.
  \bibinfo{pages}{795--804}.
\newblock


\bibitem[\protect\citeauthoryear{Chen, Wang, and Wang}{Chen
  et~al\mbox{.}}{2010a}]%
        {chen2010scalable}
\bibfield{author}{\bibinfo{person}{Wei Chen}, \bibinfo{person}{Chi Wang}, {and}
  \bibinfo{person}{Yajun Wang}.} \bibinfo{year}{2010}\natexlab{a}.
\newblock \showarticletitle{Scalable influence maximization for prevalent viral
  marketing in large-scale social networks}. In
  \bibinfo{booktitle}{\emph{Proceedings of the 16th ACM SIGKDD international
  conference on Knowledge discovery and data mining}}.
  \bibinfo{pages}{1029--1038}.
\newblock


\bibitem[\protect\citeauthoryear{Chen, Wu, and Yu}{Chen et~al\mbox{.}}{2020}]%
        {chen2020scalable}
\bibfield{author}{\bibinfo{person}{Wei Chen}, \bibinfo{person}{Ruihan Wu},
  {and} \bibinfo{person}{Zheng Yu}.} \bibinfo{year}{2020}\natexlab{}.
\newblock \showarticletitle{Scalable lattice influence maximization}.
\newblock \bibinfo{journal}{\emph{IEEE Transactions on Computational Social
  Systems}} \bibinfo{volume}{7}, \bibinfo{number}{4} (\bibinfo{year}{2020}),
  \bibinfo{pages}{956--970}.
\newblock


\bibitem[\protect\citeauthoryear{Chen, Yuan, and Zhang}{Chen
  et~al\mbox{.}}{2010b}]%
        {chen2010scal}
\bibfield{author}{\bibinfo{person}{Wei Chen}, \bibinfo{person}{Yifei Yuan},
  {and} \bibinfo{person}{Li Zhang}.} \bibinfo{year}{2010}\natexlab{b}.
\newblock \showarticletitle{Scalable influence maximization in social networks
  under the linear threshold model}. In \bibinfo{booktitle}{\emph{2010 IEEE
  international conference on data mining}}. IEEE, \bibinfo{pages}{88--97}.
\newblock


\bibitem[\protect\citeauthoryear{Domingos and Richardson}{Domingos and
  Richardson}{2001}]%
        {domingos2001mining}
\bibfield{author}{\bibinfo{person}{Pedro Domingos} {and} \bibinfo{person}{Matt
  Richardson}.} \bibinfo{year}{2001}\natexlab{}.
\newblock \showarticletitle{Mining the network value of customers}. In
  \bibinfo{booktitle}{\emph{Proceedings of the seventh ACM SIGKDD international
  conference on Knowledge discovery and data mining}}. \bibinfo{pages}{57--66}.
\newblock


\bibitem[\protect\citeauthoryear{Golovin and Krause}{Golovin and
  Krause}{2011}]%
        {golovin2011adaptive}
\bibfield{author}{\bibinfo{person}{Daniel Golovin} {and}
  \bibinfo{person}{Andreas Krause}.} \bibinfo{year}{2011}\natexlab{}.
\newblock \showarticletitle{Adaptive submodularity: Theory and applications in
  active learning and stochastic optimization}.
\newblock \bibinfo{journal}{\emph{Journal of Artificial Intelligence Research}}
   \bibinfo{volume}{42} (\bibinfo{year}{2011}), \bibinfo{pages}{427--486}.
\newblock


\bibitem[\protect\citeauthoryear{Gottschalk and Peis}{Gottschalk and
  Peis}{2015}]%
        {gottschalk2015submodular}
\bibfield{author}{\bibinfo{person}{Corinna Gottschalk} {and}
  \bibinfo{person}{Britta Peis}.} \bibinfo{year}{2015}\natexlab{}.
\newblock \showarticletitle{Submodular function maximization on the bounded
  integer lattice}. In \bibinfo{booktitle}{\emph{International Workshop on
  Approximation and Online Algorithms}}. Springer, \bibinfo{pages}{133--144}.
\newblock


\bibitem[\protect\citeauthoryear{Guo, Chen, and Wu}{Guo et~al\mbox{.}}{2020}]%
        {guo2020continuous}
\bibfield{author}{\bibinfo{person}{Jianxiong Guo}, \bibinfo{person}{Tiantian
  Chen}, {and} \bibinfo{person}{Weili Wu}.} \bibinfo{year}{2020}\natexlab{}.
\newblock \showarticletitle{Continuous Activity Maximization in Online Social
  Networks}.
\newblock \bibinfo{journal}{\emph{IEEE Transactions on Network Science and
  Engineering}} (\bibinfo{year}{2020}), \bibinfo{pages}{1--1}.
\newblock
\urldef\tempurl%
\url{https://doi.org/10.1109/TNSE.2020.2993042}
\showDOI{\tempurl}


\bibitem[\protect\citeauthoryear{Guo and Wu}{Guo and Wu}{2019}]%
        {guo2019novel}
\bibfield{author}{\bibinfo{person}{Jianxiong Guo} {and} \bibinfo{person}{Weili
  Wu}.} \bibinfo{year}{2019}\natexlab{}.
\newblock \showarticletitle{A novel scene of viral marketing for complementary
  products}.
\newblock \bibinfo{journal}{\emph{IEEE Transactions on Computational Social
  Systems}} \bibinfo{volume}{6}, \bibinfo{number}{4} (\bibinfo{year}{2019}),
  \bibinfo{pages}{797--808}.
\newblock


\bibitem[\protect\citeauthoryear{Guo and Wu}{Guo and Wu}{2020a}]%
        {guo2020influence}
\bibfield{author}{\bibinfo{person}{Jianxiong Guo} {and} \bibinfo{person}{Weili
  Wu}.} \bibinfo{year}{2020}\natexlab{a}.
\newblock \showarticletitle{Influence Maximization: Seeding Based on Community
  Structure}.
\newblock \bibinfo{journal}{\emph{ACM Transactions on Knowledge Discovery from
  Data}} \bibinfo{volume}{14}, \bibinfo{number}{6} (\bibinfo{year}{2020}),
  \bibinfo{pages}{66:1--66:22}.
\newblock


\bibitem[\protect\citeauthoryear{Guo and Wu}{Guo and Wu}{2020b}]%
        {guo2020k}
\bibfield{author}{\bibinfo{person}{Jianxiong Guo} {and} \bibinfo{person}{Weili
  Wu}.} \bibinfo{year}{2020}\natexlab{b}.
\newblock \showarticletitle{A k-hop collaborate game model: Adaptive strategy
  to maximize total revenue}.
\newblock \bibinfo{journal}{\emph{IEEE Transactions on Computational Social
  Systems}} \bibinfo{volume}{7}, \bibinfo{number}{4} (\bibinfo{year}{2020}),
  \bibinfo{pages}{1058--1068}.
\newblock


\bibitem[\protect\citeauthoryear{Han, Huang, Xiao, Tang, Sun, and Tang}{Han
  et~al\mbox{.}}{2018}]%
        {han2018efficient}
\bibfield{author}{\bibinfo{person}{Kai Han}, \bibinfo{person}{Keke Huang},
  \bibinfo{person}{Xiaokui Xiao}, \bibinfo{person}{Jing Tang},
  \bibinfo{person}{Aixin Sun}, {and} \bibinfo{person}{Xueyan Tang}.}
  \bibinfo{year}{2018}\natexlab{}.
\newblock \showarticletitle{Efficient algorithms for adaptive influence
  maximization}.
\newblock \bibinfo{journal}{\emph{Proceedings of the VLDB Endowment}}
  \bibinfo{volume}{11}, \bibinfo{number}{9} (\bibinfo{year}{2018}),
  \bibinfo{pages}{1029--1040}.
\newblock


\bibitem[\protect\citeauthoryear{Hatano, Fukunaga, and Kawarabayashi}{Hatano
  et~al\mbox{.}}{2016}]%
        {hatano2016adaptive}
\bibfield{author}{\bibinfo{person}{Daisuke Hatano}, \bibinfo{person}{Takuro
  Fukunaga}, {and} \bibinfo{person}{Ken-Ichi Kawarabayashi}.}
  \bibinfo{year}{2016}\natexlab{}.
\newblock \showarticletitle{Adaptive Budget Allocation for Maximizing Influence
  of Advertisements.}. In \bibinfo{booktitle}{\emph{IJCAI}}.
  \bibinfo{pages}{3600--3608}.
\newblock


\bibitem[\protect\citeauthoryear{Huang, Tang, Han, Xiao, Chen, Sun, Tang, and
  Lim}{Huang et~al\mbox{.}}{2020}]%
        {huang2020efficient}
\bibfield{author}{\bibinfo{person}{Keke Huang}, \bibinfo{person}{Jing Tang},
  \bibinfo{person}{Kai Han}, \bibinfo{person}{Xiaokui Xiao},
  \bibinfo{person}{Wei Chen}, \bibinfo{person}{Aixin Sun},
  \bibinfo{person}{Xueyan Tang}, {and} \bibinfo{person}{Andrew Lim}.}
  \bibinfo{year}{2020}\natexlab{}.
\newblock \showarticletitle{Efficient Approximation Algorithms for Adaptive
  Influence Maximization}.
\newblock \bibinfo{journal}{\emph{The VLDB Journal}} (\bibinfo{year}{2020}).
\newblock


\bibitem[\protect\citeauthoryear{Kempe, Kleinberg, and Tardos}{Kempe
  et~al\mbox{.}}{2003}]%
        {kempe2003maximizing}
\bibfield{author}{\bibinfo{person}{David Kempe}, \bibinfo{person}{Jon
  Kleinberg}, {and} \bibinfo{person}{{\'E}va Tardos}.}
  \bibinfo{year}{2003}\natexlab{}.
\newblock \showarticletitle{Maximizing the spread of influence through a social
  network}. In \bibinfo{booktitle}{\emph{Proceedings of the ninth ACM SIGKDD
  international conference on Knowledge discovery and data mining}}. ACM,
  \bibinfo{pages}{137--146}.
\newblock


\bibitem[\protect\citeauthoryear{Leskovec and Krevl}{Leskovec and
  Krevl}{2014}]%
        {snapnets}
\bibfield{author}{\bibinfo{person}{Jure Leskovec} {and} \bibinfo{person}{Andrej
  Krevl}.} \bibinfo{year}{2014}\natexlab{}.
\newblock \bibinfo{title}{{SNAP Datasets}: {Stanford} Large Network Dataset
  Collection}.
\newblock \bibinfo{howpublished}{\url{http://snap.stanford.edu/data}}.
\newblock


\bibitem[\protect\citeauthoryear{Nguyen, Thai, and Dinh}{Nguyen
  et~al\mbox{.}}{2016}]%
        {nguyen2016stop}
\bibfield{author}{\bibinfo{person}{Hung~T Nguyen}, \bibinfo{person}{My~T Thai},
  {and} \bibinfo{person}{Thang~N Dinh}.} \bibinfo{year}{2016}\natexlab{}.
\newblock \showarticletitle{Stop-and-stare: Optimal sampling algorithms for
  viral marketing in billion-scale networks}. In
  \bibinfo{booktitle}{\emph{Proceedings of the 2016 International Conference on
  Management of Data}}. \bibinfo{pages}{695--710}.
\newblock


\bibitem[\protect\citeauthoryear{Richardson and Domingos}{Richardson and
  Domingos}{2002}]%
        {richardson2002mining}
\bibfield{author}{\bibinfo{person}{Matthew Richardson} {and}
  \bibinfo{person}{Pedro Domingos}.} \bibinfo{year}{2002}\natexlab{}.
\newblock \showarticletitle{Mining knowledge-sharing sites for viral
  marketing}. In \bibinfo{booktitle}{\emph{Proceedings of the eighth ACM SIGKDD
  international conference on Knowledge discovery and data mining}}.
  \bibinfo{pages}{61--70}.
\newblock


\bibitem[\protect\citeauthoryear{Rossi and Ahmed}{Rossi and Ahmed}{2015}]%
        {nr}
\bibfield{author}{\bibinfo{person}{Ryan~A. Rossi} {and}
  \bibinfo{person}{Nesreen~K. Ahmed}.} \bibinfo{year}{2015}\natexlab{}.
\newblock \showarticletitle{The Network Data Repository with Interactive Graph
  Analytics and Visualization}. In \bibinfo{booktitle}{\emph{AAAI}}.
\newblock
\urldef\tempurl%
\url{http://networkrepository.com}
\showURL{%
\tempurl}


\bibitem[\protect\citeauthoryear{Soma and Yoshida}{Soma and Yoshida}{2015}]%
        {soma2015generalization}
\bibfield{author}{\bibinfo{person}{Tasuku Soma} {and} \bibinfo{person}{Yuichi
  Yoshida}.} \bibinfo{year}{2015}\natexlab{}.
\newblock \showarticletitle{A generalization of submodular cover via the
  diminishing return property on the integer lattice}. In
  \bibinfo{booktitle}{\emph{Advances in Neural Information Processing
  Systems}}. \bibinfo{pages}{847--855}.
\newblock


\bibitem[\protect\citeauthoryear{Soma and Yoshida}{Soma and Yoshida}{2017}]%
        {soma2017non}
\bibfield{author}{\bibinfo{person}{Tasuku Soma} {and} \bibinfo{person}{Yuichi
  Yoshida}.} \bibinfo{year}{2017}\natexlab{}.
\newblock \showarticletitle{Non-monotone dr-submodular function maximization}.
  In \bibinfo{booktitle}{\emph{Thirty-First AAAI conference on artificial
  intelligence}}.
\newblock


\bibitem[\protect\citeauthoryear{Soma and Yoshida}{Soma and Yoshida}{2018}]%
        {soma2018maximizing}
\bibfield{author}{\bibinfo{person}{Tasuku Soma} {and} \bibinfo{person}{Yuichi
  Yoshida}.} \bibinfo{year}{2018}\natexlab{}.
\newblock \showarticletitle{Maximizing monotone submodular functions over the
  integer lattice}.
\newblock \bibinfo{journal}{\emph{Mathematical Programming}}
  \bibinfo{volume}{172}, \bibinfo{number}{1-2} (\bibinfo{year}{2018}),
  \bibinfo{pages}{539--563}.
\newblock


\bibitem[\protect\citeauthoryear{Sun, Huang, Yu, and Chen}{Sun
  et~al\mbox{.}}{2018}]%
        {sun2018multi}
\bibfield{author}{\bibinfo{person}{Lichao Sun}, \bibinfo{person}{Weiran Huang},
  \bibinfo{person}{Philip~S Yu}, {and} \bibinfo{person}{Wei Chen}.}
  \bibinfo{year}{2018}\natexlab{}.
\newblock \showarticletitle{Multi-round influence maximization}. In
  \bibinfo{booktitle}{\emph{Proceedings of the 24th ACM SIGKDD International
  Conference on Knowledge Discovery \& Data Mining}}.
  \bibinfo{pages}{2249--2258}.
\newblock


\bibitem[\protect\citeauthoryear{Tang, Huang, Xiao, Lakshmanan, Tang, Sun, and
  Lim}{Tang et~al\mbox{.}}{2019}]%
        {tang2019efficient}
\bibfield{author}{\bibinfo{person}{Jing Tang}, \bibinfo{person}{Keke Huang},
  \bibinfo{person}{Xiaokui Xiao}, \bibinfo{person}{Laks~VS Lakshmanan},
  \bibinfo{person}{Xueyan Tang}, \bibinfo{person}{Aixin Sun}, {and}
  \bibinfo{person}{Andrew Lim}.} \bibinfo{year}{2019}\natexlab{}.
\newblock \showarticletitle{Efficient approximation algorithms for adaptive
  seed minimization}. In \bibinfo{booktitle}{\emph{Proceedings of the 2019
  International Conference on Management of Data}}.
  \bibinfo{pages}{1096--1113}.
\newblock


\bibitem[\protect\citeauthoryear{Tang, Tang, Xiao, and Yuan}{Tang
  et~al\mbox{.}}{2018}]%
        {tang2018online}
\bibfield{author}{\bibinfo{person}{Jing Tang}, \bibinfo{person}{Xueyan Tang},
  \bibinfo{person}{Xiaokui Xiao}, {and} \bibinfo{person}{Junsong Yuan}.}
  \bibinfo{year}{2018}\natexlab{}.
\newblock \showarticletitle{Online processing algorithms for influence
  maximization}. In \bibinfo{booktitle}{\emph{Proceedings of the 2018
  International Conference on Management of Data}}. \bibinfo{pages}{991--1005}.
\newblock


\bibitem[\protect\citeauthoryear{Tang and Yuan}{Tang and Yuan}{2020}]%
        {tang2020influence}
\bibfield{author}{\bibinfo{person}{Shaojie Tang} {and} \bibinfo{person}{Jing
  Yuan}.} \bibinfo{year}{2020}\natexlab{}.
\newblock \showarticletitle{Influence maximization with partial feedback}.
\newblock \bibinfo{journal}{\emph{Operations Research Letters}}
  \bibinfo{volume}{48}, \bibinfo{number}{1} (\bibinfo{year}{2020}),
  \bibinfo{pages}{24--28}.
\newblock


\bibitem[\protect\citeauthoryear{Tang, Shi, and Xiao}{Tang
  et~al\mbox{.}}{2015}]%
        {tang2015influence}
\bibfield{author}{\bibinfo{person}{Youze Tang}, \bibinfo{person}{Yanchen Shi},
  {and} \bibinfo{person}{Xiaokui Xiao}.} \bibinfo{year}{2015}\natexlab{}.
\newblock \showarticletitle{Influence maximization in near-linear time: A
  martingale approach}. In \bibinfo{booktitle}{\emph{Proceedings of the 2015
  ACM SIGMOD International Conference on Management of Data}}.
  \bibinfo{pages}{1539--1554}.
\newblock


\bibitem[\protect\citeauthoryear{Tang, Xiao, and Shi}{Tang
  et~al\mbox{.}}{2014}]%
        {tang2014influence}
\bibfield{author}{\bibinfo{person}{Youze Tang}, \bibinfo{person}{Xiaokui Xiao},
  {and} \bibinfo{person}{Yanchen Shi}.} \bibinfo{year}{2014}\natexlab{}.
\newblock \showarticletitle{Influence maximization: Near-optimal time
  complexity meets practical efficiency}. In
  \bibinfo{booktitle}{\emph{Proceedings of the 2014 ACM SIGMOD international
  conference on Management of data}}. \bibinfo{pages}{75--86}.
\newblock


\bibitem[\protect\citeauthoryear{Tong and Wang}{Tong and Wang}{2020}]%
        {tong2020adaptive}
\bibfield{author}{\bibinfo{person}{Guangmo Tong} {and} \bibinfo{person}{Ruiqi
  Wang}.} \bibinfo{year}{2020}\natexlab{}.
\newblock \showarticletitle{On Adaptive Influence Maximization under General
  Feedback Models}.
\newblock \bibinfo{journal}{\emph{IEEE Transactions on Emerging Topics in
  Computing}} (\bibinfo{year}{2020}), \bibinfo{pages}{1--1}.
\newblock
\urldef\tempurl%
\url{https://doi.org/10.1109/TETC.2020.3031057}
\showDOI{\tempurl}


\bibitem[\protect\citeauthoryear{Yuan and Tang}{Yuan and Tang}{2017}]%
        {yuan2017no}
\bibfield{author}{\bibinfo{person}{Jing Yuan} {and} \bibinfo{person}{Shaojie
  Tang}.} \bibinfo{year}{2017}\natexlab{}.
\newblock \showarticletitle{No time to observe: adaptive influence maximization
  with partial feedback}. In \bibinfo{booktitle}{\emph{Proceedings of the 26th
  International Joint Conference on Artificial Intelligence}}.
  \bibinfo{pages}{3908--3914}.
\newblock


\end{thebibliography}

\end{document}